%% file: mainUWFGraph.tex
\RequirePackage{fix-cm}
\documentclass[11pt, a4]{article} 
\usepackage[ margin=1.0in]{geometry}
\usepackage{graphicx}
\usepackage[utf8]{inputenc}
\usepackage{fancyhdr}
\usepackage{xspace}
\usepackage{tabularx}
\usepackage{algorithm}
\usepackage{algorithmicx}
\usepackage{algcompatible}
\usepackage{algpseudocode}
\algrenewcommand\algorithmicindent{0.8em}%
\usepackage[toc,page]{appendix}
\usepackage{color}
\usepackage[font=scriptsize]{caption}
\usepackage{ragged2e}
\usepackage{tabularx}
\usepackage{amsmath,amssymb,latexsym} 
\usepackage{fancyvrb}
\usepackage{float}
\usepackage{cite}
\usepackage{epstopdf}
\usepackage{epsfig}
\usepackage{graphicx}
\usepackage{pgfplots}
\usepackage{lipsum}
\usepackage{booktabs,chemformula}
\usepackage{tabularx,booktabs} 
\usepackage{etoolbox}
\usepackage{breqn}
\usepackage{subfiles}
\usepackage{sidecap}
\usepackage{authblk}
\usepackage{soul}
\usepackage{mathtools}
\usepackage{pdfpages}
\usepackage{xspace}
\pgfplotsset{compat=1.11,
        /pgfplots/ybar legend/.style={
        /pgfplots/legend image code/.code={%
        \draw[##1,/tikz/.cd,bar width=3pt,yshift=-0.2em,bar shift=0pt]
                plot coordinates {(0cm,0.8em)};},
},
}
\usepackage{calligra}
\usepackage{calrsfs}
\DeclareMathAlphabet{\mathcalligra}{T1}{calligra}{m}{n}
\DeclareCaptionFont{tiny}{\tiny}
\algrenewcommand\alglinenumber[1]{\tiny #1:}
\usepackage{multicol}
\usepackage{comment}
\usepackage{subcaption}
\usepackage{graphicx}
\usepackage{authblk}
\usepackage{appendix}
\usepackage{xspace}
\usepackage{enumitem}
\usepackage{soul}
\usepackage{mathtools}
\usepackage{breqn}
\usepackage[hidelinks]{hyperref}
\usepackage{url}
\usepackage{listings}
\usepackage{abstract}
\usepackage{todonotes}
\input{macros}

\begin{document}

\title{An Efficient Practical Concurrent Wait-Free Unbounded  Graph%\thanks{Grants or other notes
}

%\titlerunning{Concurrent Unbounded Wait-Free Graph}        % if too long for running head

\author{
     Sathya Peri$^1$, Chandra Kiran Reddy$^2$, Muktikanta Sa$^3$ \\
      Department of Computer Science \& Engineering \\
      Indian Institute of Technology Hyderabad, India \\
      \{$^1$sathya\_p, $^2$cs15btech11012,  $^3$cs15resch11012\}@iith.ac.in
}

\date{}

\maketitle

\begin{abstract}
In this paper, we propose an efficient concurrent wait-free algorithm to construct an unbounded directed graph for shared memory architecture. To the best of our knowledge that this is the first wait-free algorithm for an unbounded directed graph where insertion and deletion of vertices and/or edges can happen concurrently. To achieve wait-freedom in a dynamic setting, threads help each other to perform the desired tasks using operator descriptors by other threads. To enhance performance,  we also developed an optimized \wf graph based on the principle of \fpsp. We also prove that all graph operations are wait-free and linearizable. We implemented our algorithms in C++ and tested its performance through several micro-benchmarks. Our experimental results show an average of $9$x improvement over the global lock-based implementation.
%In this paper, we propose an efficient concurrent wait-free algorithm to construct an unbounded directed graph for shared memory architecture which is concurrently updated by multiple threads. To the best of our knowledge, this is the first wait-free algorithm for an unbounded directed graph where insertion and deletion of vertices and/or edges happened concurrently. We are able to design a wait-free graph in a dynamic setting with helping using operator descriptors by other threads. We also extend an optimized version of the wait-free algorithm in combination with the non-blocking graph to achieve the high performance. We also prove that all graph operations are wait-free and linearizable. We implement our algorithms in C++ and test through several micro-benchmarks. Our experimental results show an average of $9$x improvement over the global lock-based implementation.

\textbf{keywords: }{concurrent data structure \and lazy-list \and directed graph \and locks \and lock-free \and wait-free \and fast-path-slow-path}
\end{abstract}

\section{Introduction}
\label{sec:intro}
\input{intro}
\section{System Model}
\label{sec:model}
\input{model}
\section{The Underlying Graph Data-structure}
\label{sec:graph-ds}

\input{graph-ds}    

\section{The Wait-free Graph Algorithm}
\label{sec:wf-algo}
\input{wf-algo}

\section{An Optimized Fast Wait-free Graph Algorithm}
\label{sec:fpsp-algo}
\input{owf-algo}

\section{Proof of Correctness and Progress Guarantee}\label{sec:proof}
\input{proof}

\section{Experiments and Analysis}
\label{sec:results}
\input{expt-results}

\section{Conclusion and Future Directions}
\label{sec:conc}
\input{conclusion}

\bibliographystyle{plain}
\bibliography{biblio}

\newpage
%\appendix
%\input{appendix}
\input{code/code6.tex}

\end{document}

%% file: macros.tex
\newcommand{\punt}[1]{}
\newcommand{\cmnt}[1]{}

\newcommand{\cgds} {{concurrent graph data-structure\xspace}\xspace}

\newcommand{\cds} {CDS\xspace}

\newcommand{\lble} {linearizable\xspace}
\newcommand{\lbty} {linearizability\xspace}

\newtheorem{theorem}{Theorem}

\newcounter{history}

\newenvironment{proof}[1][Proof]{\noindent\textbf{#1.} }{\hfill $\Box$\\[0.4mm]} %\rule{0.5em}{0.5em}\\}
\newcommand{\secref}[1]{Section~\ref{sec:#1}}
\newcommand{\figref}[1]{Figure~\ref{fig:#1}}

\newcommand{\lineref}[1]{Line~\ref{lin:#1}}

\newcommand{\ignore}[1]{}

\newcommand{\operation} {op\xspace}
\newcommand{\optyp} {opType\xspace}
\newcommand{\mth} {method\xspace}
\newcommand{\cc} {correctness-criterion\xspace}

\newcommand{\lp} {LP\xspace}

\newcommand{\lylist} {lazy list\xspace}
\newcommand{\nblist} {non-blocking list\xspace}

\newcommand{\tru} {\texttt{true}\xspace}
\newcommand{\fal} {\texttt{false}\xspace}
\newcommand{\nul} {\texttt{NULL}\xspace}

\newcommand{\vnodes} {{\tt VNodes}\xspace}
\newcommand{\enodes} {{\tt ENodes}\xspace}
\newcommand{\enode}{{\tt ENode}\xspace}
\newcommand{\vnode}{{\tt VNode}\xspace}

\newcommand{\vlist} {vertex-list\xspace}
\newcommand{\elist} {edge-list\xspace}
\newcommand{\elists} {edge-lists\xspace}
\newcommand{\vh}{\texttt{VHead}\xspace}
\newcommand{\vt}{\texttt{VTail}\xspace}
\newcommand{\eh}{\texttt{EHead}\xspace}
\newcommand{\et}{\texttt{ETail}\xspace}

\newcommand{\taddv}{\texttt{addvertex}\xspace}
\newcommand{\tremv}{\texttt{remvertex}\xspace}
\newcommand{\tadde}{\texttt{addedge}\xspace}
\newcommand{\treme}{\texttt{remedge}\xspace}
\newcommand{\tconv}{\texttt{findvertex}\xspace}
\newcommand{\tcone}{\texttt{findedge}\xspace}

\newcommand{\addv}{\textsc{AddVertex}\xspace}
\newcommand{\remv}{\textsc{RemoveVertex}\xspace}

\newcommand{\createe} {\textsc{CEnode}\xspace}
\newcommand{\createv}{\textsc{CVnode}\xspace}

\newcommand{\fadd}{FetchAndAdd\xspace}

\newcommand{\nbk}{non-blocking\xspace}

\newcommand{\cas}{compare-and-swap\xspace}
\newcommand{\CAS}{\texttt{CAS}\xspace}
\newcommand{\faa}{fetch-and-add\xspace}
\newcommand{\FAA}{\texttt{FAA}\xspace}
\newcommand{\ds}{data-structure\xspace}

\newcommand{\lf}{lock-free\xspace}

\newcommand{\wf}{wait-free\xspace}
\newcommand{\wfdm}{wait-freedom\xspace}

\newcommand{\nremv}{\textsc{RemoveVertex}\xspace}
\newcommand{\naddv}{\textsc{AddVertex}\xspace}
\newcommand{\nconv}{\textsc{ConVertex}\xspace}
\newcommand{\ncone}{\textsc{ConEdge}\xspace}
\newcommand{\nadde}{\textsc{AddEdge}\xspace}
\newcommand{\nreme}{\textsc{RemoveEdge}\xspace}
\newcommand{\wfconv}{\textsc{WFConV}\xspace}
\newcommand{\wfcone}{\textsc{WFConE}\xspace}

\newcommand{\wlocv}{\textsc{WFLocV}\xspace}

\newcommand{\nwloce}{\textsc{LocE}\xspace}
\newcommand{\locuv}{\textsc{LocateUV}\xspace}
\newcommand{\wfremv}{\textsc{WFRemV}\xspace}
\newcommand{\wfaddv}{\textsc{WFAddV}\xspace}
\newcommand{\wfadde}{\textsc{WFAddE}\xspace}
\newcommand{\wfreme}{\textsc{WFRemE}\xspace}

\newcommand{\haddv}{\textsc{HelpAddV}\xspace}
\newcommand{\hremv}{\textsc{HelpRemV}\xspace}

\newcommand{\hconv}{\textsc{HelpConV}\xspace}
\newcommand{\wohv}{\textsc{ConV}\xspace}
\newcommand{\wohe}{\textsc{ConE}\xspace}
\newcommand{\hadde}{\textsc{HelpAddE}\xspace}
\newcommand{\hreme}{\textsc{HelpRemE}\xspace}
\newcommand{\hcone}{\textsc{HelpConE}\xspace}
\newcommand{\hlcv}{\textsc{HelpLocV}\xspace}
\newcommand{\hgds}{\textsc{HelpGphDS}\xspace}
\newcommand{\mxph}{\textsc{MaxPhase}\xspace}
\newcommand{\datastruct}{data-structure\xspace}
\newcommand{\enext}{{\tt enext}\xspace}
\newcommand{\vnext}{{\tt vnext}\xspace}
\newcommand{\phase}{{\tt phase}\xspace}
\newcommand{\optype}{{\tt opType}\xspace}
\newcommand{\type}{{\tt type}\xspace}

\newcommand{\vkey}{{\tt vkey}\xspace}
\newcommand{\ekey}{{\tt ekey}\xspace}
\newcommand{\vnod}{{\tt vnode}\xspace}
\newcommand{\enod}{{\tt enode}\xspace}

\newcommand{\vsrc}{{\tt vsrc}\xspace}
\newcommand{\vdest}{{\tt vdest}\xspace}
\newcommand{\ODA}{{\tt ODA}\xspace}
\newcommand{\state}{\texttt{state}\xspace}
\newcommand{\stat}{\texttt{state}\xspace}
\newcommand{\pointv}{{\tt pointv}\xspace}

\newcommand{\wfree}{wait-free\xspace}

\newcommand{\fpsp}{fast-path-slow-path\xspace}

\newcommand{\maxphase}{\texttt{maxph}\xspace}

\newcommand{\vntp}{{\tt VERTEX NOT PRESENT}\xspace}
\newcommand{\entp}{{\tt EDGE NOT PRESENT}\xspace}
\newcommand{\ventp}{{\tt VERTEX OR EDGE NOT PRESENT}\xspace}

\newcommand{\ep}{{\tt EDGE PRESENT}\xspace}
\newcommand{\eap}{{\tt EDGE ALREADY PRESENT}\xspace}
\newcommand{\eadd}{{\tt EDGE ADDED}\xspace}

\newcommand{\er}{{\tt EDGE REMOVED}\xspace}

\newcommand{\success}{{\tt success}\xspace}
\newcommand{\failure}{{\tt failure}\xspace}
\newcommand{\maxfail}{{\tt MAX\_FAIL}\xspace}

\newcommand{\isMarked}{\textsc{isMrkd}\xspace}
\newcommand{\MarkedRef}{\textsc{MrkdRf}\xspace}
\newcommand{\unMarkedRef}{\textsc{UnMrkdRf}\xspace}

%% file: intro.tex
%Graphs are ubiquitous. These graphs are represented as pairwise relationships among objects along with their properties, with the objects as vertices and relationships as edges. Because of their tremendous functionality, graphs are being used in various research areas such as social networking(linkedin, facebook, twitter,  google+, quora, etc.), semantic, data mining, image processing, VLSI design, road network, graphics, blockchains and many more. Generally, these graphs are extensively \textit{large} and \textit{dynamic} in nature, that is, they undergo dynamic changes like addition and removal of vertices and/or edges\cite{Demetrescu+:DynGraph::book:2004}. These applications also need an efficient \ds which supports dynamic changes and can expand at run-time depending on the availability of memory in the machine.

Graphs are very useful structures that have a wide variety of applications. They usually represented as pairwise relationships among objects, with the objects as vertices and relationships as edges. They are applicable in various research areas such as social networking (facebook, twitter, etc.), semantic, data mining, image processing, VLSI design, road network, graphics, blockchains and many more.
 
In many of these applications, the graphs are very \textit{large} and \textit{dynamic} in nature, that is, they undergo changes over time like the addition and removal of vertices and/or edges\cite{Demetrescu+:DynGraph::book:2004}. Hence, to precisely model these applications, we need an efficient \ds which supports dynamic changes and can expand at run-time depending on the availability of memory in the machine.

Nowadays, with multicore systems becoming ubiquitous \emph{concurrent data-structures (\cds)} \cite{Maurice+:AMP:book:2012} have become popular. \cds such as concurrent stacks, queues, hash-tables, etc. allow multiple threads to operate on them concurrently while maintaining correctness, \emph{\lbty} \cite{Herlihy+:lbty:TPLS:1990}. These structures can efficiently harness the power multi-core systems. Thus, a multi-threaded concurrent unbounded graph data-structure can effectively model dynamic graphs as described above. 

The \cc for \cds is \lbty \cite{Herlihy+:lbty:TPLS:1990} which ensures that the affect of every \mth seems to take place somewhere at some atomic step between the invocation and response of the \mth. The atomic step referred to as a \emph{linearization point (\lp)}.  Coming to progress conditions, a \mth of a \cds is \emph{\wf} \cite{Herlihy+:OnNatProg:opodis:2001, Maurice+:AMP:book:2012} if it ensures that the \mth finishes its execution in a finite number of steps. A \emph{\lf} \cds ensures, at least one of its \mth{s} is guaranteed to complete in a finite number of steps. A \lf algorithm never enters deadlocks but can possibly starve. On the other hand, \wf algorithms are starvation-free. In many of the \wf and \lf algorithms proposed in the literature, threads help each other to achieve the desired tasks. 

%To handle the performance issues, the idea of a multicore system is proposed recently, where multiple threads were built into a single system. Due to the emergence of such computational power of these systems, it is essential to design efficient data-structures to represent the dynamic graph such that multiple threads can access concurrently without losing the correctness. There have been various efforts introduced over the last decades, to build concurrent data structures like stacks, queues, sets, trees

%Due to more complex or irregular morphology, the concurrent graph data-structures and the related operations have not explored in the past, and we are not aware of any such \ds except \nbk graph which is recently proposed by Chatterjee et al.\cite{Chatterjee+:NbGraph:ICDCN-19}. Applications relying on graphs mostly use a sequential implementation and access to the shared data-structures synchronized through the global locks, which causes severe performance bottlenecks.

Concurrent graph data-structures have not been explored in detail in the literature with some work appearing recently \cite{Kallimanis+:WFGraph:opodis:2015,Chatterjee+:NbGraph:ICDCN-19}. Applications relying on graphs mostly use a sequential implementation while some parallel implementations synchronize using global locks which causes severe performance bottlenecks.

% and the related operations \todo{why practical?} 

In this paper, we describe an efficient practical concurrent \wf unbounded directed graph (for shared memory system) which supports concurrent insertion/deletion of vertices and edges while ensuring \lbty \cite{Herlihy+:lbty:TPLS:1990}. The algorithm for \wf concurrent graph \ds is based on the \nbk graph by Chatterjee et al. \cite{Chatterjee+:NbGraph:ICDCN-19} and \wf algorithm proposed by Timnat et al.\cite{Timnat:WFLis:opodis:2012}. Our implementation is not a straightforward extension to lock-free/wait-free list implementation but has several non-trivial key supplements. This can be seen from the \lp{s} of edge methods which in many cases lie outside their method and depend on other graph operations running concurrently. We believe the design of the graph data-structure is such that it can help to identify other useful properties on a graph such as reachability, cycle detection, shortest path, betweenness centrality, diameter, etc. 

%To enhance performance,  we also developed an optimized \wf graph based on the principle of \fpsp \cite{Kogan+:fpsp:ppopp:2012}. The basic idea is that the lock-free algorithms are fast as compare to the wait-free algorithms in practice. as they don't require helping always. So, instead of helping always, check whether help indeed required or not. 

To enhance performance,  we also developed an optimized \wf graph based on the principle of \fpsp \cite{Kogan+:fpsp:ppopp:2012}. The basic idea is that the lock-free algorithms are fast as compare to the wait-free algorithms in practice. So, instead of always executing in the \wf manner (slow-path), threads normally execute \mth{s} in the \lf manner (fast-path). If a thread executing a \mth in the \lf manner, fails to complete in a certain threshold number of iterations, switches to \wf execution and eventually terminates.

\subsection{Contributions}
In this paper, we present an efficient practical concurrent \wf unbounded directed graph \ds. The main contributions of our work are summarized below:
\begin{enumerate}
	\item We describe an Abstract Data Type (ADT) that maintains a \wf directed graph $G = (V,E)$. It comprises of the following \mth{s} on the sets $V$ and $E$: (1) Add Vertex: $\wfaddv$ (2) Remove Vertex: $\wfremv$, (3) Contains Vertex: $\wfconv$ (4) Add Edge: $\wfadde$ (5) Remove Edge: $\wfreme$ and (6) Contains Edge: $\wfcone$. The \wf graph is represented as an adjacency list similar in \cite{Chatterjee+:NbGraph:ICDCN-19} (\secref{graph-ds}). 
	
	\item We implemented the directed graph in a dynamic setting with threads helping each other using operator descriptors to achieve \wfdm (\secref{wf-algo}).
	
	\item We also extended the \wf graph to enhance the performance and achieve a fast \wf graph based on the principle of \fpsp proposed by Kogan et al.\cite{Kogan+:fpsp:ppopp:2012} (\secref{fpsp-algo}). %\todo{We must explain what this property is and why is it important}.
	
	\item Formally, we prove for the correctness by showing the operations of the concurrent graph \ds are \lble \cite{Herlihy+:lbty:TPLS:1990}. We also prove the \wf progress guarantee of the operations $\wfaddv$, $\wfremv$, $\wfconv$, $\wfadde$, $\wfreme$, and $\wfcone$ (\secref{proof}).
	
	\item  We evaluated the \wf algorithms in C++ implementation and tested through several micro-benchmarks. Our experimental results show on an average of $9$x improvement over the sequential and global lock implementation (\secref{results}).
\end{enumerate}
%\vspace{-0.2in}
% based on practical rather than theoretical considerations

\subsection{Related Work}
Kallimanis and Kanellou \cite{Kallimanis+:WFGraph:opodis:2015} presented a concurrent graph that supports \wf edge updates and traversals. They represented the graph using adjacency matrix, with an upper bound on number of vertices. As a result, their graph data-structure does not allow any insertion or deletion of vertices after initialization of the graph. Although this might be useful for some applications such as road networks, this may not be adequate for many real-world applications which need dynamic modifications of vertices as well as unbounded graph size. 

A recent work by Chatterjee et al. \cite{Chatterjee+:NbGraph:ICDCN-19} proposed a non-blocking \cgds which allowed multiple threads to perform dynamic insertion and deletion of vertices and/or edges in \lf manner. Our paper extends their \ds while ensuring that all the graph operations are \wf. %However, their algorithm does not allow the graph operations to achieve wait-freedom.

%A recent work by Chatterjee et al. \cite{Chatterjee+:NbGraph:ICDCN-19} proposed a non-blocking \cgds which allowed multiple threads to perform dynamic insertion and deletion of vertices and/or edges in \lf manner. Our paper extends their \ds to maintain all the graph operations to be \wf  of a directed graph. However, their algorithm does not allow the graph operations to achieve wait-freedom.

%% file: model.tex
\vspace{1mm}
\noindent
\textbf{The Memory Model.} We consider an asynchronous shared-memory model with a finite set of $p$ processors accessed by a finite set of $n$ threads. The threads communicate with each other by invoking atomic operations on the shared objects such as atomic \texttt{read}, \texttt{write}, \texttt{\faa} (\FAA)  and \texttt{\cas} (\CAS) instructions. 

%\textbf{The Memory Model.} We consider an asynchronous shared-memory model with a finite set of $p$ processors accessed by a finite set of $n$ threads. The threads communicate with each other by invoking atomic operations on the shared objects. We execute our \wf graph \ds on a shared-memory multi-core with multi-threading enabled which supports atomic \texttt{read}, \texttt{write}, \texttt{\faa} (\FAA)  and \texttt{\cas} (\CAS) instructions. 

%For a concurrent non-blocking implementation of the \ds, we consider a \textit{shared-memory system} consisting of a finite set of \textit{processors} accessed by a finite set of \textit{threads} that run in a completely asynchronous manner. The threads communicate with each other by invoking \mth{s} on shared objects and getting corresponding responses. The pointers and other fields of the various nodes are implemented by the shared objects. The system supports atomic \texttt{read}, \texttt{write}, \texttt{\faa} (\FAA)  and \texttt{\cas} (\CAS) instructions. 
%\input{code/code1}
%\noindent

An $\FAA{(x, a)}$ instruction atomically increments the value at the memory location ${x}$ by the value $a$. Similarly, a \texttt{CAS}${(x, a, a')}$ is an atomic instruction that checks if the current value at a memory location ${x}$ is equivalent to the given value ${a}$, and only if true, changes the value of ${x}$ to the new value ${a'}$ and returns \tru; otherwise the memory location remains unchanged and the instruction returns \fal. Such a system can be perfectly realized by a Non-Uniform Memory Access (NUMA) computer with one or more multi-processor CPUs.

\vspace{1mm}
\noindent
\textbf{Correctness.} We consider \textit{\lbty} introduced by Herlihy \& Wing \cite{Herlihy+:lbty:TPLS:1990} as the correctness criterion for the graph operations. We assume that the execution generated by a \ds is a collection of \mth invocation and response events. Each invocation of a method call has a subsequent response. An execution is \lble if it is possible to assign an atomic event as a \emph{linearization point} (\emph{\lp}) inside the execution interval of each \mth such that the result of each of these \mth{s} is the same as it would be in a sequential execution in which the \mth{s} are ordered by their \lp{s} \cite{Herlihy+:lbty:TPLS:1990}. 

\vspace{1mm}
\noindent
\textbf{Progress.} The progress properties specify when a thread invoking operations on the shared memory objects completes in the presence of other concurrent threads. In this context, we provide the graph implementation with \mth{s} that satisfy wait-freedom, based on the definitions in Herlihy and Shavit \cite{Herlihy+:OnNatProg:opodis:2001}. A \mth of a \cds is \wf if it completes in finite number of steps. A \ds implementation is \wf if all its \mth{s} are \wf. This ensures per-thread progress and is the most reliable non-blocking progress guarantee in a concurrent system. A \ds is \lf if its \mth{s} get invoked by multiple concurrent threads, then one of them will complete in finite number of steps. 

%\textbf{Progress.} The progress properties specify when a thread invoking operations on the shared memory objects completes in the presence of other concurrent threads. In this context, we provide the graph implementation with operations that satisfy wait-freedom, based on the definitions in Herlihy and Shavit \cite{Herlihy+:OnNatProg:opodis:2001}. A \ds implementation is \wf if it ensures that in each thread finishes its execution in a finite number of steps. This ensures pre-thread progress of the whole system and is the most reliable non-blocking progress guarantee in a concurrent system. 

%% file: graph-ds.tex
 In this section, we give a detailed construction of the graph \ds, which is a combination of non-blocking graph based on \cite{Chatterjee+:NbGraph:ICDCN-19} and \wf construction based on \cite{Timnat:WFLis:opodis:2012, Timnat+:WFDS:ppopp:2014}. We represent the concurrent directed graph as an adjacency list representations. Hence, it is constructed as a collection set of vertices stored as linked-list manner wherein each vertex also holds a list of neighboring vertices which it has outgoing edges. 

%The problem addressed in this paper is as follows: A concurrent directed graph $G = (V,E)$, which is dynamically being modified by a fixed set of concurrent threads. In this setting, threads may perform insertion / deletion of vertices or edges to the graph. 

%The data structure, we consider a \textit{shared-memory system} consisting of a finite set of \textit{processors} accessed by a finite set of \textit{threads} that run in a completely asynchronous manner. The threads communicate with each other by invoking \mth{s} on shared objects and getting corresponding responses. The pointers and other fields of the various nodes are implemented by the shared objects. The system supports atomic \texttt{read}, \texttt{write}, \texttt{\faa} (\FAA)  and \texttt{\cas} (\CAS) instructions. 

\begin{figure}[!t]
	\captionsetup{font=footnotesize}
	\begin{footnotesize}
			\begin{tabbing}
		%	\hspace{0.1in} \= 		\hspace{0.1in} \=
		\hspace{0.1in} \=  \hspace{0.1in} \= \\
			\> {\bf class \vnode \{} \\
			\> \> \texttt{int}  $\vkey$; // immutable key field \\
			\> \> {\vnode~ \vnext;} // atomic refe., pointer to the next \vnode\\
			\> \> {\enode~ \enext;} // atomic ref., pointer to the \elist \\   %the \eh of
			\> \} 
			\\
			\> {\bf class \enode \{} \\
			\> \> \texttt{int}  $\ekey$;   //  immutable key field\\
			\> \> \texttt{\vnode}  $\pointv$; // pointer from the \enode to its \vnode.  \\
			\> \> {\enode~ \enext;} // atomic ref., pointer to the next \enode \\
			\> \}\\
			\> {\bf class \ODA \{} \\
			\> \> \texttt{unsigned long}  $\phase$; // phase number of each operation.\\
			%\> \> \texttt{int} \hspace{0.25in} $key$;\\
			%\> \> \texttt{int} \hspace{0.25in} $\key$;\\
			\> \> \texttt{\optype} \hspace{0.09in} $\type$; // type of the operation.\\
			\> \> \texttt{\vnode}  \hspace{0.18in}$\vnod$; // pointer to the \vnode. \\
			\> \> {\enode~ \hspace{0.15in}$\enod$;} // pointer to the \enode.\\
			\> \> \texttt{\vnode} \hspace{0.15in} $\vsrc,\vdest$; // pointer to the source and destination \vnode \\ %for an \enode.  \\
			\> \}\\
			\\
			\> {\vnode \vh, \vt}; // Sentinel nodes for the \vlist \\
			\> {unsigned long \maxphase}; // atomic variable which keeps track of op. number. \\
			\> {\ODA} {\state[]}; // global state array for posting operations, array size is \\
			\hspace{1.0in}same as number of threads\\
		\end{tabbing}
		\vspace{-0.2in}
	\end{footnotesize}
	\setlength{\belowcaptionskip}{-10pt}
	\caption{Structure of \enode, \vnode and \ODA.}
	\label{fig:struct-evnode}
\end{figure}
%The \vnode and \enode structure are similar to \cite{Chatterjee+:NbGraph:ICDCN-19}, and an \ODA is similar to \cite{Timnat:WFLis:opodis:2012,Timnat+:WFDS:ppopp:2014} with some modifications.
The \vnode, \enode and \ODA structures depicted in \figref{struct-evnode}.  The \vnode consists of two atomic pointers \vnext and \enext and an immutable key \vkey. The \vnext points to the next \vnode in the \vlist, whereas, \enext points to the head of the \elist which is the list of outgoing neighboring vertices. Similarly, an \enode also has an atomic pointer \enext  points to the next \enode in the \elist and a pointer \pointv points to the corresponding \vnode. Which helps direct access to its \vnode while doing any traversal like BFS, DFS and also helps to delete the incoming edges, detail regarding this described in \secref{wf-algo}.   We assume that all the vertices in the \vlist have unique identification key which captured by \vkey field. Similarly, all the edge nodes for a vertex in the \elist have unique identification key which captured by \ekey field. 

Besides \vnodes and \enodes, we also have an \state array with an operation-descriptor(\ODA) for each thread. Each thread's \state entry recounts its current state. An \ODA composed of six fields: a phase number
\phase, an operation type \type indicates the current operation executed by this thread, possible operation types given at the end of this section, a pointer to the vertex node \vnod, which used for any vertex operations, a pointer to the edge node \enod, which is used for any edge operations, and a pair of \vnodes pointers \vsrc and \vdest, used for any edge operations to store the source and destination \vnodes of an edge.

\begin{figure}[H]
%\captionsetup{font=scriptsize}
	\centerline{\scalebox{0.75}{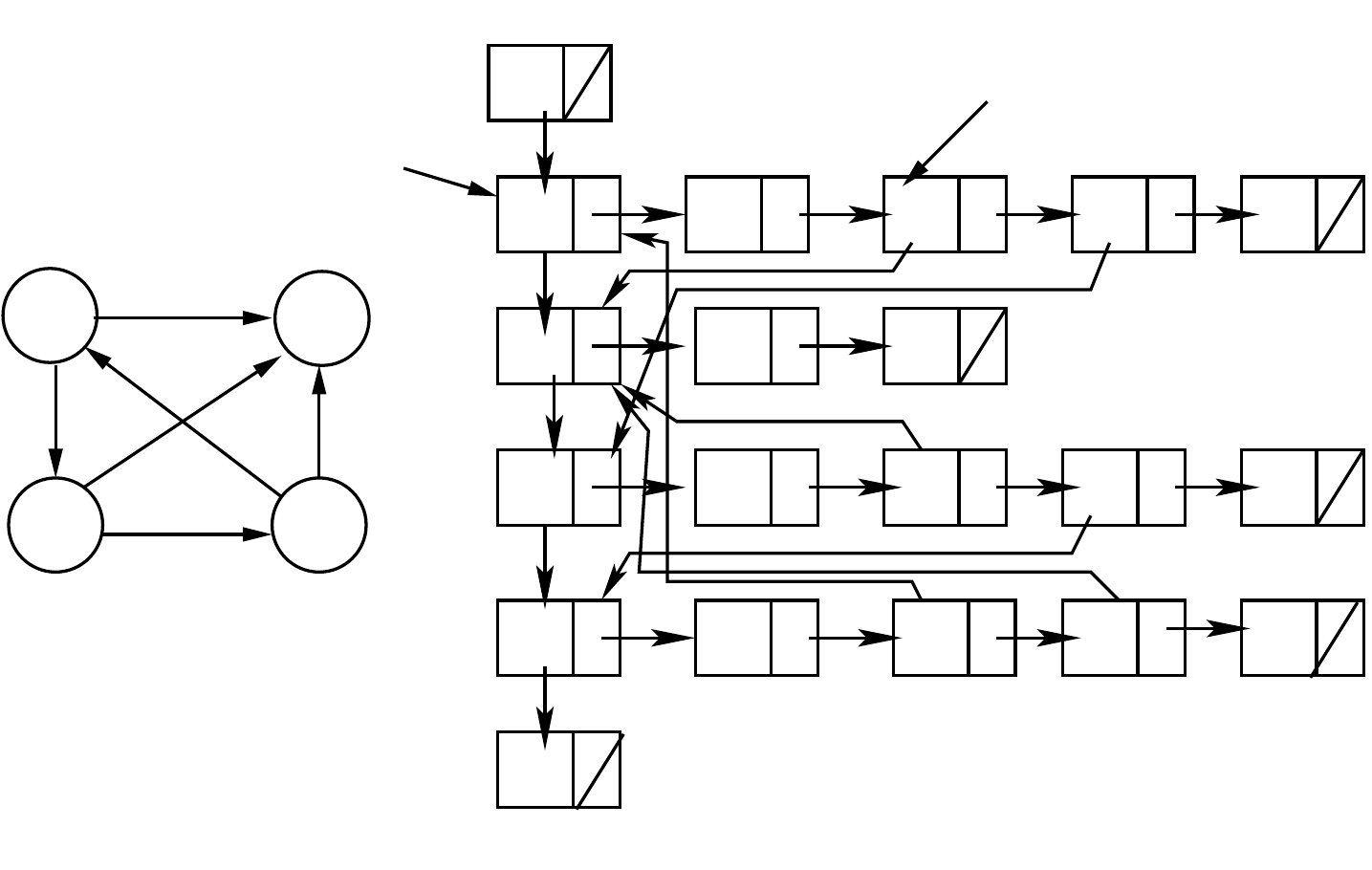}}
		\setlength{\belowcaptionskip}{-10pt}
	\caption{(a) A directed Graph (b) The \wf graph representation for (a).}
	\label{fig:conGraph}
\end{figure}

Our wait-free concurrent directed graph \ds supports six major operations: \wfaddv, \wfremv, \wfconv, \wfadde, \wfreme, and \wfcone. We use the helping mechanism like \cite{Timnat:WFLis:opodis:2012, Kogan+:WFQue:ppopp:2011, Timnat+:WFDS:ppopp:2014} to achieve the wait-freedom for all our graph operations. Before begin to execute an operation, a thread starts invoking a special state array \state similar as Timnat et al. \cite{Timnat:WFLis:opodis:2012}. This \state shared among the threads. All threads can see the  details of the operation they are running during their execution. Whenever an operation starts execution it publishes its operation in the \state array, then all other threads try to help to finish its execution. When the operation finished its execution, the outcome result is also announced to the \state array, using a \CAS, which substitutes the old existing \type to the new one. 

We initialize the \vlist with dummy head(\vh) and tail(\vt) (called sentinels) with values $\text{-}\infty$ and $\infty$ respectively. Similarly, each \elists is also initialized with a dummy head(\eh) and tail(\et)(see \figref{conGraph}).

Our \wf graph \ds maintains some \textit{invariants}: (a) the \vlist is sorted based on the \vnode's key value $\vkey$ and each unmarked \vnode is reachable from the \vh, and (b) also each \elists are sorted based on the \enode's key value $\ekey$ and unmarked \enodes are reachable from the \eh of the corresponding \vnode.

\subsection{The Abstract Data Type(ADT)}
A \wf graph is defined as a directed graph $G = (V, E)$, where $V$ is the set of vertices and $E$ is the set of directed edges. Each edge in $E$ is an ordered pair of vertices belonging to $V$. A vertex $v$ $\in V$ has an immutable unique key $\vkey$ denoted by $v(\vkey)$. A directed edge from the vertex $v(\ekey_1)$ to $v(\ekey_2)$ is denoted as $e(v(\ekey_1), v(\ekey_2))$ $\in E$. For simplicity, we denote $e(v(\ekey_1), v(\ekey_2))$ as $e(\ekey_2)$, which means the $v(\ekey_1)$ has a neighbouring vertex $v(\ekey_2)$. 
We also defined an ADT for operations on $G$ which are exported by the \wf graph \ds, are given bellow:

\begin{enumerate}
\item The $\wfaddv(\vkey)$ operation adds a vertex $v(\vkey)$ to the graph, only if $v(\vkey) \notin V$ and then returns \tru, otherwise it returns \fal. 
%\vspace{-1mm}
\item The $\wfremv(\vkey)$ operation deletes a vertex $v(\vkey)$ from $V$, only if $v(\vkey) \in V$ and then returns \tru, otherwise it returns \fal. Once a vertex $v(\vkey)$ is deleted successfully all its outgoing and incoming edges also removed.
%\vspace{-1mm}
\item The $\wfconv(\vkey)$ operation returns \tru, if $v(\vkey) \in V$; otherwise returns \fal.
%\vspace{-1mm}
\item The $\wfadde(\ekey_1, \ekey_2)$ adds an edge $e(v(\ekey_1), v(\ekey_2))$ to $E$, only if $e(v(\ekey_1)$, $v(\ekey_2))$ $\notin$ $E$ and $v(\ekey_1) \in V$ and $v(\ekey_2) \in V$ then it returns \eadd. If $v(\ekey_1) \notin V$ or $v(\ekey_2) \notin V$, it returns \vntp. If $e(v(\ekey_1), v(\ekey_2)) \in E$, it returns \eap.	
\item The $\wfreme(\ekey_1, \ekey_2)$ deletes the edge $e(v(\ekey_1)$, $v(\ekey_2))$ from $E$, only if $e(v(\ekey_1)$, $v(\ekey_2))$ $\in E$ and $v(\ekey_1)$ $\in V$ and $v(\ekey_2)$ $\in V$ then it returns \er. If $v(\ekey_1)$ $\notin V$ or $v(\ekey_2)$ $\notin V$, it returns \vntp. If $e(v(\ekey_1)$, $v(\ekey_2))$ $\notin E$, it returns \entp.
\item The $\wfcone(\ekey_1,\ekey_2)$ if $e(v(\ekey_1)$, $v(\ekey_2))$ $\in E$ and $v(\ekey_1) \in V$ and $v(\ekey_2) \in V$ then it returns \ep, otherwise it returns \ventp.
\item The $\hgds(\phase)$ operation ensures that each thread completes its own operation and helps in completing all the pending operations with lower \phase numbers.
\end{enumerate}
%\vspace{-3mm}
When the operations \wfaddv, \wfremv, \wfconv, \wfadde, \wfreme, and \wfcone start execution they will get a new \phase number and post their operation in the \state array and then invoke the $\hgds(\phase)$
(Help Graph \ds)operation. All the helping methods check if their \phase number is the same as the thread's \phase number in the \state array, otherwise, they return \fal.

%The possible operation type for the \ODA state values are: (1). \taddv, requested for help to insert a \vnode into the \vlist, (2). \tremv, requested for help to delete a \vnode from the \vlist, (3). \tconv, requested for help to find a \vnode, from the \vlist, (4). \tadde, requested for help to insert an \enode into an \elist, (5). \treme, requested for help to delete the \enode from an \elist, (6). \tcone, requested for help to find the \enode from an \elist, (7). \success, if any of the graph operations finished its execution successfully, and (8). \failure, if any of the graph operation unable to finish its execution.

%\ignore{
The possible operation type for the \ODA state values are:
\begin{enumerate}
    \item \taddv: requested for help to insert a \vnode into the \vlist
    \item \tremv: requested for help to delete a \vnode from the \vlist.
    \item \tconv : requested for help to find a \vnode, from the \vlist.
   \item \tadde: requested for help to insert an \enode into an \elist.
    \item \treme: requested for help to delete the \enode from an \elist.
    \item \tcone : requested for help to find the \enode from an \elist.
    \item \success: if any of the graph operations finished its execution successfully.
    \item \failure: if any of the graph operation unable to finish its execution.
\end{enumerate}
%}
The first six states used to ask for help from other threads whereas the last two states  thread does not ask for any help from other threads.

% For the vnodes $u$ and $v$, we represent a particular $\enode$ $e$ as the edge $(u,v)$, from $u$ to $v$, such that $e.val$ = $v.val$. We assume that a typical application invokes significantly more contains \mth{s} (\textit{\cone} and \textit{\conv}) than the update \mth{s} (\textit{\addv, \remv, \adde, \reme}).
%\noindent Table \ref{tab:seq-spe} describes the sequential specification of each method formally in any given global state $S$ before the execution of the method and future state $S'$ after executing it sequentially. The \textit{Pre-state} is the shared state before $\inv$ event and the \textit{Post-state} is also the shared state just after the $\rsp$ event of a method, which is depicted in the Figure \ref{fig:exec/hist}.

%% file: figs/conGraph.pdf_t
\begin{picture}(0,0)%
\includegraphics{figs/conGraph.pdf}%
\end{picture}%
\setlength{\unitlength}{4144sp}%
\begingroup\makeatletter\ifx\SetFigFont\undefined%
\gdef\SetFigFont#1#2#3#4#5{%
  \reset@font\fontsize{#1}{#2pt}%
  \fontfamily{#3}\fontseries{#4}\fontshape{#5}%
  \selectfont}%
\fi\endgroup%
\begin{picture}(6544,4159)(3334,-4004)
\put(7651,-916){\makebox(0,0)[lb]{\smash{{\SetFigFont{12}{14.4}{\rmdefault}{\bfdefault}{\updefault}{\color[rgb]{0,0,0}$2$}%
}}}}
\put(5806,-2941){\makebox(0,0)[lb]{\smash{{\SetFigFont{12}{14.4}{\rmdefault}{\bfdefault}{\updefault}{\color[rgb]{0,0,0}$4$}%
}}}}
\put(5806,-3571){\makebox(0,0)[lb]{\smash{{\SetFigFont{12}{14.4}{\rmdefault}{\bfdefault}{\updefault}{\color[rgb]{0,0,0}$+\infty$}%
}}}}
\put(5941,-3931){\makebox(0,0)[lb]{\smash{{\SetFigFont{12}{14.4}{\rmdefault}{\bfdefault}{\updefault}{\color[rgb]{0,0,0}$\vt$}%
}}}}
\put(8551,-916){\makebox(0,0)[lb]{\smash{{\SetFigFont{12}{14.4}{\rmdefault}{\bfdefault}{\updefault}{\color[rgb]{0,0,0}$3$}%
}}}}
\put(9406,-916){\makebox(0,0)[lb]{\smash{{\SetFigFont{12}{14.4}{\rmdefault}{\bfdefault}{\updefault}{\color[rgb]{0,0,0}$+\infty$}%
}}}}
\put(9406,-601){\makebox(0,0)[lb]{\smash{{\SetFigFont{12}{14.4}{\rmdefault}{\bfdefault}{\updefault}{\color[rgb]{0,0,0}$\et$}%
}}}}
\put(7651,-1546){\makebox(0,0)[lb]{\smash{{\SetFigFont{12}{14.4}{\rmdefault}{\bfdefault}{\updefault}{\color[rgb]{0,0,0}$+\infty$}%
}}}}
\put(8506,-2221){\makebox(0,0)[lb]{\smash{{\SetFigFont{12}{14.4}{\rmdefault}{\bfdefault}{\updefault}{\color[rgb]{0,0,0}$4$}%
}}}}
\put(9361,-2941){\makebox(0,0)[lb]{\smash{{\SetFigFont{12}{14.4}{\rmdefault}{\bfdefault}{\updefault}{\color[rgb]{0,0,0}$+\infty$}%
}}}}
\put(9361,-2221){\makebox(0,0)[lb]{\smash{{\SetFigFont{12}{14.4}{\rmdefault}{\bfdefault}{\updefault}{\color[rgb]{0,0,0}$+\infty$}%
}}}}
\put(3511,-2446){\makebox(0,0)[lb]{\smash{{\SetFigFont{12}{14.4}{\rmdefault}{\bfdefault}{\updefault}{\color[rgb]{0,0,0}$3$}%
}}}}
\put(5761,-286){\makebox(0,0)[lb]{\smash{{\SetFigFont{12}{14.4}{\rmdefault}{\bfdefault}{\updefault}{\color[rgb]{0,0,0}$-\infty$}%
}}}}
\put(6706,-916){\makebox(0,0)[lb]{\smash{{\SetFigFont{12}{14.4}{\rmdefault}{\bfdefault}{\updefault}{\color[rgb]{0,0,0}$-\infty$}%
}}}}
\put(6751,-1546){\makebox(0,0)[lb]{\smash{{\SetFigFont{12}{14.4}{\rmdefault}{\bfdefault}{\updefault}{\color[rgb]{0,0,0}$-\infty$}%
}}}}
\put(6751,-2221){\makebox(0,0)[lb]{\smash{{\SetFigFont{12}{14.4}{\rmdefault}{\bfdefault}{\updefault}{\color[rgb]{0,0,0}$-\infty$}%
}}}}
\put(6751,-2941){\makebox(0,0)[lb]{\smash{{\SetFigFont{12}{14.4}{\rmdefault}{\bfdefault}{\updefault}{\color[rgb]{0,0,0}$-\infty$}%
}}}}
\put(5851,-916){\makebox(0,0)[lb]{\smash{{\SetFigFont{12}{14.4}{\rmdefault}{\bfdefault}{\updefault}{\color[rgb]{0,0,0}$1$}%
}}}}
\put(5761,-16){\makebox(0,0)[lb]{\smash{{\SetFigFont{12}{14.4}{\rmdefault}{\bfdefault}{\updefault}{\color[rgb]{0,0,0}$\vh$}%
}}}}
\put(6706,-601){\makebox(0,0)[lb]{\smash{{\SetFigFont{12}{14.4}{\rmdefault}{\bfdefault}{\updefault}{\color[rgb]{0,0,0}$\eh$}%
}}}}
\put(3511,-1411){\makebox(0,0)[lb]{\smash{{\SetFigFont{12}{14.4}{\rmdefault}{\bfdefault}{\updefault}{\color[rgb]{0,0,0}$1$}%
}}}}
\put(4816,-1456){\makebox(0,0)[lb]{\smash{{\SetFigFont{12}{14.4}{\rmdefault}{\bfdefault}{\updefault}{\color[rgb]{0,0,0}$2$}%
}}}}
\put(4771,-2446){\makebox(0,0)[lb]{\smash{{\SetFigFont{12}{14.4}{\rmdefault}{\bfdefault}{\updefault}{\color[rgb]{0,0,0}$4$}%
}}}}
\put(7111,-3616){\makebox(0,0)[lb]{\smash{{\SetFigFont{12}{14.4}{\rmdefault}{\bfdefault}{\updefault}{\color[rgb]{0,0,0}$(b)$}%
}}}}
\put(3826,-3571){\makebox(0,0)[lb]{\smash{{\SetFigFont{12}{14.4}{\rmdefault}{\bfdefault}{\updefault}{\color[rgb]{0,0,0}$(a)$}%
}}}}
\put(5806,-1546){\makebox(0,0)[lb]{\smash{{\SetFigFont{12}{14.4}{\rmdefault}{\bfdefault}{\updefault}{\color[rgb]{0,0,0}$2$}%
}}}}
\put(5806,-2221){\makebox(0,0)[lb]{\smash{{\SetFigFont{12}{14.4}{\rmdefault}{\bfdefault}{\updefault}{\color[rgb]{0,0,0}$3$}%
}}}}
\put(7651,-2221){\makebox(0,0)[lb]{\smash{{\SetFigFont{12}{14.4}{\rmdefault}{\bfdefault}{\updefault}{\color[rgb]{0,0,0}$2$}%
}}}}
\put(8056,-421){\makebox(0,0)[lb]{\smash{{\SetFigFont{12}{14.4}{\rmdefault}{\bfdefault}{\updefault}{\color[rgb]{0,0,0}$\enode$}%
}}}}
\put(4951,-601){\makebox(0,0)[lb]{\smash{{\SetFigFont{12}{14.4}{\rmdefault}{\bfdefault}{\updefault}{\color[rgb]{0,0,0}$\vnode$}%
}}}}
\put(8461,-2941){\makebox(0,0)[lb]{\smash{{\SetFigFont{12}{14.4}{\rmdefault}{\bfdefault}{\updefault}{\color[rgb]{0,0,0}$2$}%
}}}}
\put(7696,-2941){\makebox(0,0)[lb]{\smash{{\SetFigFont{12}{14.4}{\rmdefault}{\bfdefault}{\updefault}{\color[rgb]{0,0,0}$1$}%
}}}}
\end{picture}%

%% file: wf-algo.tex
In this section, we present the technical details of all \wf graph operations. The design of \wf graph \ds based on the adjacency list representation. Hence, it is implemented as a collection (list) of vertices wherein each vertex holds a list of vertices to which it has outgoing edges. The implementation is a linked list of \vnodes and \enodes as shown in \figref{conGraph}. The implementation of each of these lists based on the \nbk graph \cite{Chatterjee+:NbGraph:ICDCN-19} and  \wf  construction based on \cite{Timnat:WFLis:opodis:2012, Timnat+:WFDS:ppopp:2014} and \nblist concurrent-set \cite{Fomitchev+:LFList:podc:2004, Harris:NBList:disc:2001, Michael:LFHashList:spaa:2002, Valois:LFList:podc:1995}. The \wf graph algorithm depicted in \figref{wf-methods1}, \ref{fig:wf-methods2}, \ref{fig:wf-methods3}, and \ref{fig:wf-methods5}.

%We represent the directed graph as adjacency list, which is the combination of \nbk graph based on \cite{Chatterjee+:NbGraph:ICDCN-19} and \wf  construction based on \cite{Timnat:WFLis:opodis:2012, Timnat+:WFDS:ppopp:2014}. It is the collection set of vertices stored as list, \vlist, and each vertex's neighboring vertices are also stored in a list, \elist. The \wf graph algorithm is depicted in the \figref{ac-v-methods}, \ref{fig:ac-e-methods}, \ref{fig:ac-rcbl1-methods}, \ref{fig:ac-rcbl2-methods} and \ref{fig:e2-methods}.
\noindent \textbf{Pseudo-code convention:}
We use $p.x$ to access the member field $x$ of a class object pointer $p$.
To return multiple variables from an operation we use $\langle x_1, x_2,\ldots,x_n \rangle$. To avoid the overhead of another field in the node structure, we use bit-manipulation: last one significant bit of a pointer $p$. In case of an x86\text{-}64 bit architecture, memory has a 64-bit boundary and the last three least significant bits are unused. So, we use the last one significant bit of the pointer. We define three methods $\isMarked(p)$ return \tru if last significant bit of pointer $p$ is set to $1$, else, it returns \fal, $\MarkedRef(p)$, $\unMarkedRef(p)$ sets last significant bit of the pointer $p$ to $1$ and $0$ respectively. An invocation of $\createv(\vkey)$ creates a new \vnode with key $\vkey$. Similarly, an invocation of $\createe(\ekey)$ creates a new \enode with key $\ekey$. For a newly created \vnode and \enode the pointer fields are initialised with \nul value.% Similarly, a newly created \enode initialises its pointer fields to \nul as well. 

\input{code/code1.tex}

\subsection{The Helping Procedure}
When an operation is invoked by a thread to start its execution, at first, it chooses a unique \phase number which is the higher than all previously chosen phase numbers by other threads. The main objective of assigning a unique \phase number is to help the operations with lower \phase number. This means whenever a thread started its execution with a new \phase number, it tries to help other unfinished operations whose \phase number is lower. Which allows all operations to get help to finish their execution to ensures the starvation-free. The phase selection procedure, in Line  \ref{mxphstart} to \ref{mxphend}, executed by reading the current phase number and then atomically increments the \maxphase, using an \FAA. Once the \phase numbers is chosen, the thread publishes the operation in the \state array by updating its entry. Then it invokes the \hgds(in Line \ref{hgdsstart} to \ref{hgdsend}) procedure where it
traverses through the \state array and tries to help the operations whose \phase number is lower than or equal to its, which ensures the unfinished operations gets help from other threads to finish the execution. This ensures wait-freedom.

\input{code/code2.tex}

\subsection{The Vertex Methods}
\label{sec:working-con-graph-methods:add}
The \wf vertex operations $\wfaddv$, $\wfremv$, and $\wfconv$ depicted in \figref{wf-methods1} and their corresponding helping procedures \haddv, \hremv, and \hconv depicted in \figref{wf-methods2}, and \ref{fig:wf-methods3}. If the vertex set keys are finite(up to available memory in the system), we also have an optimized case where the $\wfconv$ neither helps nor accepts any help from other threads. This because to achieve higher throughput, we assume $\wfconv$ is called more frequently than the $\wfaddv$ and $\wfremv$ operations, so it does not allow any help. Without being affected by each other, all the vertex operations are \wf.

%by selecting a phase number, allocating a new node with the input key, and installing a link to it in the state array.

A $\wfaddv(\vkey)$ operation invoked by passing the  $\vkey$ to be inserted, in Line \ref{wfaddvstart} to \ref{wfaddvend}. A \wfaddv operation starts by choosing a \phase number, creating a new \vnode by invoking $\createv(\vkey)$, posting its operation on the \state array. Then the thread calls $\hgds(\phase)$ to invoke the helping mechanism. After that, it traverses the \state array and helps all pending operations and tries to complete its operation. In the next step the same thread(or a helping thread) enters $\haddv(\phase)$ and verifies the \phase number and type of operation \optype, if they match with \taddv then it invokes $\hlcv(\phase)$ to traverse the \vlist until it finds a vertex with its key greater than or equal to $v(\vkey)$. In the process of traversal, it physically deletes all logically deleted \vnodes, using a \CAS. Once it reaches the appropriate location checks whether the $\vkey$ is already present. If the $\ekey$ is not present earlier it attempts a \CAS to add the new \vnode(\lineref{helpaddv-succ-phyaddv-cas}). On an unsuccessful \CAS, it retries. After the operation completes $\haddv(\phase)$ \success or \failure reported to the \state array. 
\input{code/code3.tex}

In the process of traversal, it is possible that the threads which are helping might have been inserted the $v(\vkey)$ but not publish \success. Also, it is possible that the $v(\vkey)$ we are trying to insert was already inserted and then removed by some other thread and then a different \vnode with \vkey was inserted to the \vlist. We properly handled these cases.

To identify these cases, we check the \vnodes that was discovered during the process of traversal. If that is the same as $v(\vkey)$ that we are trying to
insert, then we reported \success to the \state array(see \lineref{helpaddv-succ-cas1}). Also, we check if that \vnode is marked by invoking \isMarked procedure for deletion, means the $v(\vkey)$ already inserted and then marked for deletion, then we also reported \success to the \state array (see \lineref{helpaddv-succ-cas2}), else we try to report \failure.
Like $\wfaddv$, a $\wfremv(\vkey)$ operation invoked by passing the  $\vkey$ to be deleted, in Line \ref{wfremvstart} to \ref{wfremvend}. It starts by choosing a \phase number, announcing its operation on the \state array to delete the $v(\vkey)$. Then the thread calls $\hgds(\phase)$ to invoke the helping mechanism. Like a $\wfaddv$ operation, it traverses the \state array and helps all pending operations and tries to complete its operation. In the next step the same thread(or a helping thread) enters $\hremv(\phase)$ and verifies the \phase number and type of operation \optype, if they match with \tremv then it invokes $\hlcv(\phase)$ to traverse the \vlist until it finds a vertex with its key greater than or equal to $v(\vkey)$. In the process of traversal, it physically deletes all logically deleted \vnodes, using a \CAS. Once it reaches the appropriate location checks whether the $\vkey$ is already present. If the $\ekey$ is present it attempts to remove the \vnode in two steps (like \cite{Harris:NBList:disc:2001}), (a) atomically marks the \vnext of current \vnode, using a \CAS(\lineref{helpremv-succ-logical-cas}), and (b) atomically updates the \vnext of the predecessor \vnode to point to the \vnext of current \vnode, using a \CAS(\lineref{helpremv-succ-phy-cas}). On any unsuccessful \CAS it will cause the operation to restart from the \hremv procedure. After the operation completes $\hremv$ will update \success to the \state array. 
A $\wfconv(\vkey$) operation is much simpler than \wfaddv and \wfremv. We have two cases based on with and without help. For the helping case a $\wfconv(\vkey$) operation, in Line \ref{wfconvstart} to \ref{wfconvend}, first starts publishing the operation in the \state array like other operations. Then any helping thread will search it in the \vlist, If the searching key is present and not been marked it reported \success, else it reported \failure, using a \CAS to the \state array. The \hconv procedure, in Line \ref{hconvstart} to \ref{hconvend}, guarantees that the list traversal does not affected from infinite insertion of \vnodes, this is because other threads will first help this operation before inserting a new \vnode.
\input{code/code4.tex}

Unlike a \haddv and \hremv, a \hconv  procedure does not need a loop to update the \state array if any failure of the \CAS. For the without helping case a $\wohv(\vkey$) operation, in Line \ref{wohvstart} to \ref{wohvend}, first traverses the \vlist in a \wf manner skipping all logically marked \vnodes until it finds a vertex with its key greater than or equal to $\vkey$. Once it reaches the appropriate \vnode, checks its key value equals to $\vkey$, and it is unmarked, then it returns \tru otherwise returns \fal. We do not allow $\wohv$ for any helping in the process of traversal. This is because if the vertex set keys are finite(upto available memory in the system) to achieve higher throughput and $\wohv$ is called more frequently than the $\wfaddv$ and $\wfremv$ operations, so it does not allow any help.

\subsection{The Edge Methods}
\label{sec:working-con-graph-methods:remove}
%\vspace{-3mm}

The \wf graph edge operations $\wfadde$, $\wfreme$, and $\wfcone$ are depicted in \figref{wf-methods1} and \ref{fig:wf-methods2} and their corresponding helping procedures $\hadde$, $\hreme$, and $\hcone$ are defined in \figref{wf-methods3}. 
\input{code/code5.tex}

A $\wfadde(\ekey_1, \ekey_2)$ operation, in Line \ref{wfaddestart} to \ref{wfaddeend}, begins by validating the presence of the $v(\ekey_1)$ and $v(\ekey_2)$ in the \vlist by invoking \locuv and are unmarked. If the validations fail, it returns \vntp. Once the validation succeeds, \wfadde operation starts by choosing a \phase number, creating a new \enode by invoking $\createe(\ekey)$, posting its operation on the \state array alone with the \vsrc and \vdest vertices. Then the thread calls $\hgds(\phase)$ to invoke the helping mechanism. After that it traverses the \state array and helps all pending operations and tries to complete its own operation. In the next step the same thread(or a helping thread) enters $\hadde(\phase)$(in Line \ref{haddestart} to \ref{haddeend}) and verifies the \phase number and type of operation \optype, if they match with the \tadde then it invokes $\nwloce(v(\ekey_1),\ekey_2)$(\lineref{hadde-newloce}) to traverse the \elist until it finds an \enode with its key greater than or equal to $\ekey_2$. In the process of traversal, it physically deletes two kinds of logically deleted \enodes,
(a) the \enodes whose \vnode is logically deleted, using a \CAS, and (b) the logically deleted \enodes, using a \CAS. Once it reaches the appropriate location checks if the $\ekey_2$ is already present or not. If present, then it attempts a \CAS to add the new \enode(\lineref{helpadde-succ-wfadde-cas}). On an unsuccessful \CAS, it retries. After the operation completes $\hadde(\phase)$ \success or \failure reported to the \state array. 

In the process of traversal in the \hadde procedure, it is possible that the threads which are helping might have been inserted the $e(\ekey_2)$ but not publish \success. Also it is possible that the $e(\vkey_2)$ we are trying to insert was already inserted and then removed by some other thread and then a different \enode with $\ekey_2$ was inserted to the \elist of $v(\ekey_1)$ . We properly handled these cases. Like \wfaddv and \wfremv, we check the \enodes that was discovered during the process of traversal is the same key $\ekey_2$ that we are trying to insert, then we reported \success to the \state array(see \lineref{helpadde-succ-cas1}). Also we check if that \enode is marked( by invoking \isMarked), means the $e(\ekey_2)$ already inserted and then marked for deletion, then we also reported \success to \state array(see \lineref{helpadde-succ-cas2}), else we try to report \failure.

In each start of the  \textbf{while}(\lineref{search-again-hadde}) it is necessary to check for the presence of vertices $\operation.\vsrc$ and $\operation.\vdest$ in \lineref{hadde-loceuv-marked} by calling \isMarked procedure. The reason explained in  \cite{Chatterjee+:NbGraph:ICDCN-19}. This is one of the several differences between an implementation trivially extending lock-free and wait-free list and concurrent \wf graph. In fact, it can be seen that if this check is not performed, then it can result in the algorithm to not be linearizable. 

%Once the vertices $u$ and $v$ have been validated to be reachable and unmarked in the vertex list, the thread traverses the \elist of vertex $u$ until an edge node with key greater than $v$ has been encountered, say $ecurr$ and it's predecessor say $epred$.  If the $ecurr$ holds a $val$ equal to the $v$ the \enode to be added, then it returns the \failure. Otherwise the \enode should be inserted between $ecurr$ and $epred$. This is done by first updating the new \enode's \enext pointer to point to the $ecurr$, and then updating the $epred$'s \enext to point to it. The later one is done using a \CAS to prevent the race condition and any failure of the \CAS will cause the operation to restart from the \hadde method. Finally,  when the new \enode has been added, the \wfadde return \success.  After the operation completes $\hadde(phase)$ will update the \state to \success or \failure. 
\ignore{
\begin{figure}[H]
%\centering
 %\floatbox[{\capbeside\thisfloatsetup{capbesideposition={right,top},capbesidewidth=4cm}}]{figure}[\FBwidth]
%\captionsetup{font=scriptsize}
	\centerline{\scalebox{0.65}{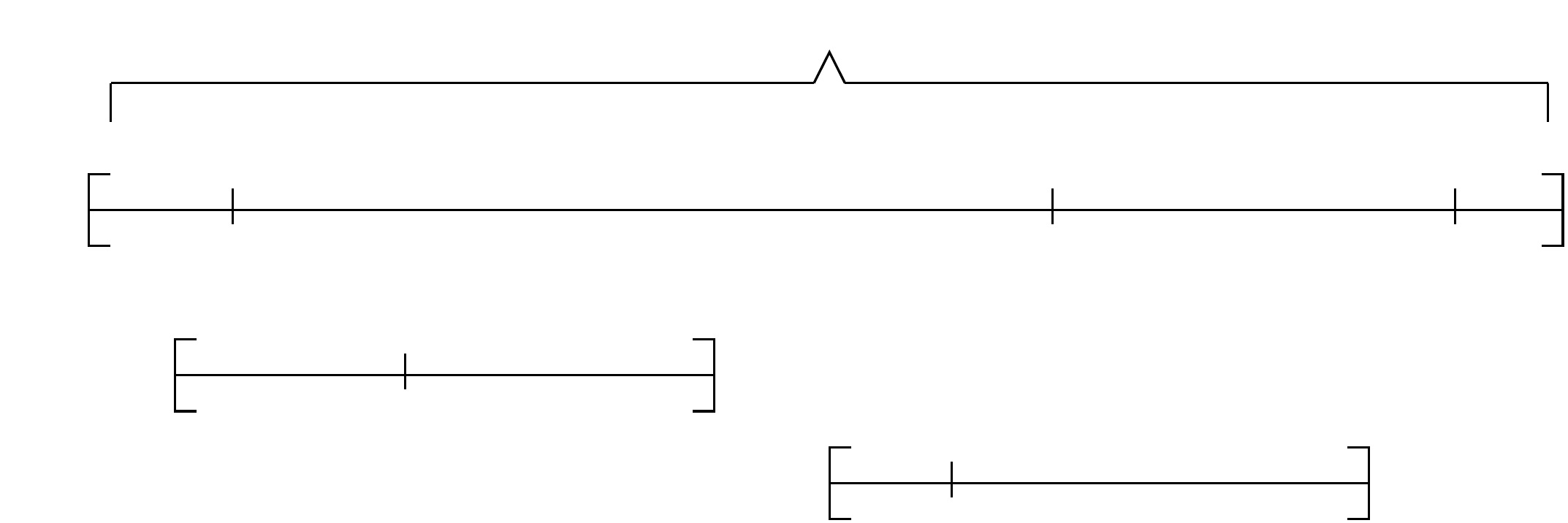}}
% \resizebox{.5\linewidth}{!}{\input{figs/false.pdf_t}}
%	\setlength{\belowcaptionskip}{-10pt}
	\caption{This figure depicts why we need an additional check to locate vertices in $\wfadde(u,v)$. A thread $T_1$ trying to perform $\wfadde(u, v, \success)$, first invokes \hlcv. Just after $T_1$ has verified vertex $u$, thread $T_2$ deletes vertex $u$. Also vertex $v$ gets added by thread $T_3$ just before $T_1$ verifies it. So, now thread $T_1$ has successfully tested for the presence of vertices $u$ and $v$ in the vertex list, and then it proceeds to add edge $(u,v)$, returning \success. However, as is evident, no possible sequentially generated history of the given concurrent execution is correct. Hence an additional check must be performed before proceeding to actually add the edge.}
	\label{fig:noseqhist}
\end{figure}
}
A $\wfreme(\ekey_1, \ekey_2)$ operation proceeds almost identical to the $\wfadde$, in Line \ref{wfremestart} to \ref{wfremeend}, begins by validating the presence of the $v(\ekey_1)$ and $v(\ekey_2)$ and are unmarked. If the validations fail, it returns \vntp. Once the validation succeeds, \wfreme operation starts by choosing a \phase number, creating a new \enode, posting its operation on the \state array alone with the \vsrc and \vdest vertices. Then it gone through helping mechanism $\hgds(\phase)$. After that it traverses the \state array and helps all pending operations and tries to complete its own operation. In the next step the same thread(or a helping thread) enters $\hreme(\phase)$(in Line \ref{hremestart} to \ref{hremeend}) and verifies the \phase number and type of operation \optype, if they match with the \treme then it traverse the \elist until it finds an \enode with its key greater than or equal to $\ekey_2$. Like \hadde, in the process of traversal, it physically deletes two kinds of logically deleted \enodes, (a) the \enodes whose \vnode is logically deleted, and (b) the logically deleted \enodes. Once it reaches the appropriate location checks if the $\ekey_2$ is already present or not, if it is, attempts to remove the $e(\ekey_2)$ in two steps like \wfremv operation, (a) atomically marks the \enext of the current \enode, using a \CAS(\lineref{cas-wfreme-lgic}), and (b) atomically updates the \enext of the predecessor \enode to point to the \enext of current \enode, using a \CAS(\lineref{cas-wfreme-phy}). On any unsuccessful \CAS, it reattempted. Like \hadde, after the operation completes $\hreme(\phase)$ \success or \failure reported to the \state array. 

Like \wfconv, a \wfcone operations also has two cases based on with and without help. For the helping case a \wfcone($\ekey_1, \ekey_2$) operation, in Line \ref{wfconestart} to \ref{wfconeend}, does similar work like \wfadde and \wfreme operation. It publishes the operation in the \state array and then invokes the $\hgds(\phase)$. If the current \enode is equal to the $\ekey_2$ and $e(\ekey_2)$ unmarked and $v(\ekey_2)$ and $v(\ekey_2)$ also unmarked, it updates the \state array to \success, using a \CAS. Otherwise \failure is updated to the \state array.

For the without helping case a $\wohe(\ekey_1, \ekey_2$) operation, in Line \ref{wohestart} to \ref{woheend}, validates the presence of the corresponding \vnodes. Then it traverses the \elist of $v(\ekey_1)$ in a \wf manner skipping all logically marked \enodes until it finds an edge with its key greater than or equal to $\ekey_2$. Once it reaches the appropriate \enode, checks its key value equals to $\ekey_2$, and $e(\ekey_2)$ is unmarked, and $v(\ekey_1)$ and $v(\ekey_2)$ are unmarked, then it returns \ep otherwise it returns \ventp.

%% file: code/code1.tex
\begin{figure*}[]
	\begin{subfigure}[t]{.52\textwidth}
	
\begin{algorithmic}[1]
%	\algrestore{adde-ac1}
%\renewcommand{\algorithmicprocedure}{\textbf{Operation}}
\scriptsize
\renewcommand{\algorithmicprocedure}{\textbf{Operation}}
	\Procedure{ $\mxph()$}{}\label{mxphstart}
	 \State{$\maxphase.\fadd(1)$;}
	 \State{$return$ $\maxphase$;}
	\EndProcedure\label{mxphend}
	\algstore{maxph}
\end{algorithmic}
\hrule

\begin{algorithmic}[1]
	\algrestore{maxph}
\renewcommand{\algorithmicprocedure}{\textbf{Operation}}
\scriptsize
	\Procedure{ \hgds($phase$)}{}\label{hgdsstart}
	    %\State{$tid \gets 0$;}
	    \For{($tid \gets 0$ $to $ $\stat.end()$)}
	    \State{$\ODA$ $desc\gets \stat[tid]$;}
	    \If{($desc.\phase \leq phase$)}
	    \If{($desc.\type = \taddv$)}
	    \State{$\haddv(tid, desc.\phase)$;}
	    \ElsIf{($desc.\type = \tremv$)}
	    \State{$\hremv(tid,desc.\phase)$;}
	    \ElsIf{($desc.\type = \tadde$)}
	    \State{$\hadde(tid,desc.\phase)$;}
	    \ElsIf{($desc.\type = \treme$)}
	    \State{$\hreme(tid,desc.\phase)$;}
	    \ElsIf{($desc.\type = \tconv$)}
	    \State{$\hconv(tid,desc.\phase)$;}
	    \ElsIf{($desc.\type = \tcone$)}
	    \State{$\hcone(tid,desc.\phase)$;}
	    \EndIf
	    \EndIf
	    %\State{$tid \gets tid + 1$};
	    \EndFor
	\EndProcedure\label{hgdsend}
	\algstore{hgds}
\end{algorithmic}	%\caption{}\label{addv}
	%\end{subfigure}% need this comment symbol to avoid overconvfull hbox
	%\begin{subfigure}{.24\textwidth}
	%\end{subfigure}
	    %\begin{subfigure}{.5\textwidth}
	    \hrule
	  
\begin{algorithmic}[1]
    \renewcommand{\algorithmicprocedure}{\textbf{Operation}}
	\algrestore{hgds}
	\scriptsize
	\Procedure{ $\wfaddv$($key$)}{}\label{wfaddvstart}
	\State{$tid \gets ThreadID.get()$};
	\State{$phase \gets\mxph()$};
	\State{$nv \gets new$  $\vnode(key)$;}
	\State{$\ODA$ $\operation \gets$ $new$ $\ODA(phase, \taddv,nv);$}
	\State{$\stat[tid] \gets \operation$;}
	\State{$\hgds(phase)$;}
	\If{$(\stat[tid].\type = \success)$}
	\State{$return$ $\tru$;}
	\Else
	\State{$return$ $\fal$;}
	\EndIf
	\EndProcedure\label{wfaddvend}
	\algstore{wfaddv}
    \end{algorithmic}
    \hrule
    
\begin{algorithmic}[1]
    \renewcommand{\algorithmicprocedure}{\textbf{Operation}}
	\algrestore{wfaddv}
	\scriptsize
	\Procedure{ $\wfremv$($key$)}{}\label{wfremvstart}
	\State{$tid \gets ThreadID.get()$};
	\State{$phase \gets \mxph()$};
	\State{$nv \gets new$  $\vnode(key)$;}
	\State{$\ODA$ $\operation \gets$ $new$ $\ODA(phase, \tremv,nv);$}
	\State{$\stat[tid] \gets \operation$;}
	\State{$\hgds(phase)$;}
	\If{($\stat[tid].\type = \success$)}
	\State{$return$ \tru;}
	\Else
	\State{$return$ \fal;}
	\EndIf
	\EndProcedure\label{wfremvend}
	\algstore{wfremv}
    \end{algorithmic}
    \hrule
   	\end{subfigure}
\begin{subfigure}[t]{.48\textwidth}
\begin{algorithmic}[1]
    \renewcommand{\algorithmicprocedure}{\textbf{Operation}}
	\algrestore{wfremv}
	\scriptsize
	\Procedure{$\wfconv$($key$)}{}\label{wfconvstart}
	\State{$tid \gets ThreadID.get()$};
	\State{$phase \gets \mxph()$};
	\State{$nv \gets new$  $\vnode(key)$;}
	\State{$\ODA$ $\operation \gets$ $new$ $\ODA(phase, \tconv,$ $nv);$}
	\State{$\stat[tid] \gets \operation$;}
	\State{$\hgds(phase)$;}
	\If{($\stat[tid].\type = \success$)}
	\State{$return$ $\tru$;}
	\Else
	\State{$return$ $\fal$;}
	\EndIf
	\EndProcedure\label{wfconvend}
	\algstore{wfconv}
\end{algorithmic}
\hrule
\begin{algorithmic}[1]
    \renewcommand{\algorithmicprocedure}{\textbf{Operation}}
	\algrestore{wfconv}
	\scriptsize
	\Procedure{$\wfadde$($key_1,key_2$)}{}\label{wfaddestart}
	\State{$tid \gets ThreadID.get()$};
	\State{$phase \gets \mxph()$};
	\State{$\langle v_1,v_2,flag\rangle \gets \locuv(key_1, key_2$);}\label{lin:wfadde-locuv}
    \If {($flag$ = $\fal$ $\bigvee$ $\isMarked(v_1)$ $\bigvee$ $\isMarked(v_2)$)}
    \State {$return$ \vntp }
    \EndIf
	%\State{$newe1 \gets new$  $enode(key_1)$;}
	\State{$ne \gets new$  $\enode(key_2)$;}
	\State{$\ODA$ $\operation$ $\gets new$ $\ODA$ $(phase,$ $\tadde,ne,$ $v_1,v_2);$}
	\State{$\stat[tid] \gets op$;} 
	\State{$\hgds(phase)$;}
	\If{($\stat[tid].type = \success$)}
	\State{$return$ \eadd;}
	\Else
	\State{$return$ \eap;}
	\EndIf
	\EndProcedure\label{wfaddeend}
	\algstore{wfadde}
\end{algorithmic}
\hrule
\begin{algorithmic}[1]
    \renewcommand{\algorithmicprocedure}{\textbf{Operation}}
	\algrestore{wfadde}
	\scriptsize
	\Procedure{$\wfreme$($key_1,key_2$)}{}\label{wfremestart}
	\State{$tid \gets ThreadID.get()$};
	\State{$phase \gets \mxph()$};
	\State{$ \langle v_1,v_2,flag\rangle \gets \locuv (key_1, key_2$);}
     \If {($flag$ = $\fal$ $\bigvee$ $\isMarked(v_1)$ $\bigvee$ $\isMarked(v_2)$)}
    \State {$return$ \vntp; }
    \EndIf
    \State{$ne \gets new$  $\enode(key_2)$;}
	\State{$\ODA$ $\operation$ $\gets new$ $\ODA$ $(phase,$ $\treme,ne,$ $v_1,v_2);$}
	\State{$\stat[tid] \gets op$;}
	\State{$\hgds(phase)$;} 
	\If{($\stat[tid].type = \success$)}
	\State{$return$ \er;}
	\Else
	\State{$return$ \entp;}
	\EndIf
	\EndProcedure\label{wfremeend}
	\algstore{wfreme}
\end{algorithmic}
\hrule
	\end{subfigure}
	\setlength{\belowcaptionskip}{-15pt}
	\caption{Pseudo-codes of \wfaddv, \wfremv, $\wfconv$ and \wfcone}\label{fig:wf-methods1}
\end{figure*}

%% file: code/code2.tex
\begin{figure*}[]
	\begin{subfigure}[t]{.48\textwidth}
	
\begin{algorithmic}[1]
    \renewcommand{\algorithmicprocedure}{\textbf{Operation}}
	\algrestore{wfreme}
	\scriptsize
	\Procedure{ $\wfcone$($key_1,key_2$)}{}\label{wfconestart}
    \State{$tid \gets ThreadID.get()$};
	\State{$phase \gets \mxph()$};
	\State{$ \langle v_1,v_2,flag\rangle \gets \locuv (key_1, key_2$);}
    \If {($flag$ = $\fal$ $\bigvee$ $\isMarked(v_1)$ $\bigvee$ $\isMarked(v_2)$)}
    \State {$return$ \vntp; }
    \EndIf
    \State{$ne \gets new$  $\enode(key_2)$;}
	\State{$\ODA$ $\operation$ $\gets new$ $\ODA$ $(phase,$ $\tcone,ne,$ $v_1,v_2);$}
	\State{$\stat[tid] \gets op$;}
	\State{$\hgds(phase)$;} 
	\If{($\stat[tid].\type = \success$)}
	\State{$return$ \ep;}
	\Else
	\State{$return$ \entp;}
	\EndIf
	\EndProcedure\label{wfconeend}
	\algstore{wfcone}
    \end{algorithmic}
    \hrule

\begin{algorithmic}[1]
    \renewcommand{\algorithmicprocedure}{\textbf{Operation}}
	\algrestore{wfcone}
	\scriptsize
	\Procedure{ $\hconv$($tid,phase$)}{}\label{hconvstart}
	\State{$\ODA$ $\operation$ $\gets \stat[tid]$;}
	\If{($\neg(op.\type = \tconv \land op.\phase = phase)$)}
	\State{$return$;}
    \EndIf
    \State{$v_1\gets op.\vnod$;}
    %\State{$vnode$ $v_2\gets v_1.next$;}
    \State{$\langle pred,curr \rangle \gets$ $\wlocv(v_1.\vkey)$;}
    \If{$(curr.\vkey = v_1.\vkey)$ $\wedge$ $(\neg\isMarked$ $(curr.\vnext))$}
	\State{$\ODA$ $succ$ $\gets new$ $\ODA$ $(phase,$ $\success);$}
    \State{($CAS(\stat[tid], op, succ)$);}
    \State{$return$;}
    \Else
	\State{$\ODA$ $fail$ $\gets new$ $\ODA$ $(phase,$ $\failure);$}
    \State{($CAS(\stat[tid], op, fail)$);}
    \State{$return$;}
    \EndIf
	\EndProcedure\label{hconvend}
	\algstore{hconv}
    \end{algorithmic}
    \hrule
   	\end{subfigure}
\begin{subfigure}[t]{.52\textwidth}
\begin{algorithmic}[1]
    \renewcommand{\algorithmicprocedure}{\textbf{Operation}}
	\algrestore{hconv}
	\scriptsize
	\Procedure{$\haddv$($tid,phase$)}{}\label{haddvstart}
		\While{($\tru$)} \label{lin:search-again-haddv}
		\State{$\ODA$ $\operation \gets \stat[tid]$;}
		\If{($\neg(\optyp = \taddv \land \operation.\phase = phase)$)}
		\State{$return$;}
        \EndIf
        \State{$\vnode$ $v_1\gets \operation.\vnod$;}
        \State{$\vnode$ $v_2\gets v_1.\vnext$;}
        \State{$\langle pred,curr \rangle \gets$ $\wlocv(v_1.\vkey)$;}
        \If{($curr.\vkey = v_1.\vkey $)}
        \If{($curr = v_1$ $\vee$ $\isMarked($ $curr.\vnext$ $))$)}
	    \State{$\ODA$ $succ$ $\gets new$ $\ODA$ $(phase,\success$ $);$}
        \If{$\CAS(\stat[tid], \operation, succ)$}\label{lin:helpaddv-succ-cas1}
        \State{$return$;}
        \EndIf
        \Else
	    \State{$\ODA$ $fail$ $\gets new$ $\ODA$ $(phase,\failure$ $);$}
        \If{($\CAS(\stat[tid], \operation, fail)$)} \label{lin:helpaddv-fail-cas1}
        \State{$return$;}
        \EndIf
        \EndIf
        \Else
        \If{($(\isMarked(v_1.\vnext))$)}
	    \State{$\ODA$ $succ$ $\gets new$ $\ODA$ $(phase,\success$ $);$}
        \If{($\CAS(\stat[tid], \operation, succ)$)}\label{lin:helpaddv-succ-cas2}
        \State{$return$;}
        \EndIf
        \EndIf
        \State{$\CAS(v_1.\vnext, v_2, curr)$;}
        \If{($\CAS(pred.\vnext, curr, v_1)$)} \label{lin:helpaddv-succ-phyaddv-cas}
	    \State{$\ODA$ $succ$ $\gets new$ $\ODA$ $(phase,\success$ $);$}
        \If{($\CAS(\stat[tid], \operation, succ)$)}\label{lin:helpaddv-succ-cas3}
        \State{$return$;}
        \EndIf
        \EndIf
        \EndIf
		\EndWhile
	\EndProcedure\label{haddvend}
	\algstore{haddv}
\end{algorithmic}
\hrule
	\end{subfigure}
	\setlength{\belowcaptionskip}{-15pt}
	\caption{Pseudo-codes of \wfaddv, \wfremv, $\wfconv$ and \wfcone}\label{fig:wf-methods2}
\end{figure*}

%% file: code/code3.tex
\begin{figure*}[]
	\begin{subfigure}[t]{.46\textwidth}
	
\begin{algorithmic}[1]
    \renewcommand{\algorithmicprocedure}{\textbf{Operation}}
	\algrestore{haddv}
	\scriptsize
	\Procedure{ $\hremv$($tid,phase$)}{}\label{hremvstart}
		\While{($\tru$)} \label{lin:search-again-hremv}
		\State{$\ODA$ $op\gets \stat[tid]$;}
		\If{($\neg(op.\type = \tremv \wedge op.$ $\phase = phase)$)}
		\State{$return$;}
        \EndIf
        \State{$\vnode$ $v_1\gets op.\vnod$;}
        %\State{$vnode$ $v_2\gets v_1.next$;}
        \State{$\langle pred,curr \rangle \gets$ $\wlocv(v_1.\vkey)$;}
        \State{$cnext \gets curr.\vnext$}
        \If{($curr.\vkey \neq v_1.\vkey $)}
	    \State{$\ODA$ $fail$ $\gets new$ $\ODA$ $(phase,$ $\failure);$}
        \If{($\CAS(\stat[tid], op, fail)$)}
        \State{$return$;}
        \EndIf
        \Else
        %\State{$Node$ $next\gets op.curr.\vnext$;}
        \If{$(\neg\CAS(curr.\vnext, cnext, \MarkedRef$ $(cnext))) $}\label{lin:helpremv-succ-logical-cas}
       	 \State{$goto$ \lineref{search-again-hremv};}
		\EndIf
		 \If{($\neg(\CAS(pred.\vnext, curr, cnext))$ $))$}\label{lin:helpremv-succ-phy-cas}
       	 \State{$goto$ \lineref{search-again-hremv};}
		\EndIf
        \State{$\ODA$ $succ$ $\gets new$ $\ODA$ $(phase,$ $\success);$}
        \If{($\CAS(\stat[tid], op, succ)$)}
        \State{$return$;}
        \EndIf
        \EndIf
        \EndWhile
	\EndProcedure\label{hremvend}
	\algstore{wfreme}
    \end{algorithmic}
    \hrule
    
\begin{algorithmic}[1]
    \renewcommand{\algorithmicprocedure}{\textbf{Operation}}
	\algrestore{wfreme}
	\scriptsize
	\Procedure{ $\wohv$($key$)}{}\label{wohvstart}
	\State {$vnode$ $v \gets read(\vh)$;}
	\While{$(v.\vkey < key)$}
	\State {$v \gets \unMarkedRef(v.\vnext)$ ;}
	\EndWhile
	\State {return $(v.\vkey = key) \land (\isMarked(v) = \fal$ $)$};
	\EndProcedure\label{wohvend}
	\algstore{wohv}
    \end{algorithmic}
    \hrule
 	\end{subfigure}
\begin{subfigure}[t]{.54\textwidth}
\begin{algorithmic}[1]
    \renewcommand{\algorithmicprocedure}{\textbf{Operation}}
	\algrestore{wohv}
	\scriptsize
	\Procedure{$\hadde$($tid,phase$)}{}\label{haddestart}
		\While{($\tru$)} \label{lin:search-again-hadde}
		\State{$\ODA$ $op\gets \stat[tid]$;}
		\If{($\neg(op.\type = \tadde \land op.\phase = phase)$)}
		\State{$return$;}
        \EndIf
        \If {$(\isMarked(op.\vsrc)$ $\vee$ $\isMarked(op.\vdest$ $))$}  \label{lin:hadde-loceuv-marked}
	    \State{$\ODA$ $fail$ $\gets new$ $\ODA$ $(phase,$ $\failure);$}
        \State{$\CAS(\stat[tid], op, fail)$;} 
        \State{$return$;}
        \EndIf
        %\State{$\enode$ $e_1\gets op.ne1$;}
        %\State{$key_1 \gets op.key_1$;}
        \State{$\vnode$ $v_1\gets \operation.\vsrc$;}
        \State{$\enode$ $e_2\gets op.\enod $, $e_3 \gets e_2.\enext$;}
        %\State{$\enode$ $e_3 \gets e_2.\enext$;}
        \State{$\langle pred,curr \rangle \gets \nwloce(v_1, e_2.\ekey)$;} \label{lin:hadde-newloce}
         \If{($curr.\ekey = e_2.\ekey $)} 
        \If{$(curr = e_2)$ $\vee$ $(\isMarked(curr.\enext$ $))$}
	    \State{$\ODA$ $succ$ $\gets new$ $\ODA$ $(phase,\success$ $);$}
        \If{($\CAS(\stat[tid], op, succ)$)}\label{lin:helpadde-succ-cas1}
        \State{$return$;}
        \EndIf
        \Else
	    \State{$\ODA$ $fail$ $\gets new$ $\ODA$ $(phase,\failure$ $);$} 
        \If{($\CAS(\stat[tid], op, fail)$)}  \label{lin:helpadde-fail-cas1}
        \State{$return$;}
        \EndIf
        \EndIf
        \Else
        \If{($\isMarked(e_2.\enext)$)}
	    \State{$\ODA$ $succ$ $\gets new$ $\ODA$ $(phase,\success$ $);$} 
        \If{($\CAS(\stat[tid], op, succ)$)}\label{lin:helpadde-succ-cas2}
        \State{$return$;}
        \EndIf
        \EndIf
        \State{$\CAS(e_2.\enext, e_3, curr)$;}
        \If{($\CAS(pred.\enext, curr, e_2)$)} \label{lin:helpadde-succ-wfadde-cas}
        \State{$e_2.\pointv \gets op.\vdest$;}
	    \State{$\ODA$ $succ$ $\gets new$ $\ODA$ $(phase,\success$ $);$} \label{lin:helpadde-succ-cas3}
       \If{($\CAS(\stat[tid], newOP, succ)$)}
        \State{$return$;}
        \EndIf
        \EndIf
        \EndIf
		\EndWhile
	
	\EndProcedure\label{haddeend}
	\algstore{hadde}
\end{algorithmic}
\hrule
	\end{subfigure}
	\setlength{\belowcaptionskip}{-15pt}
	\caption{Pseudo-codes of \wfaddv, \wfremv, $\wfconv$ and \wfcone}\label{fig:wf-methods3}
\end{figure*}

%% file: code/code4.tex
\begin{figure*}[]
	\begin{subfigure}[t]{.5\textwidth}
	
\begin{algorithmic}[1]
    \renewcommand{\algorithmicprocedure}{\textbf{Operation}}
	\algrestore{hadde}
	\scriptsize
	\Procedure{ $\hreme$($tid,phase$)}{}\label{hremestart}
		\While{($\tru$)} \label{lin:search-again-hreme}
		\State{$\ODA$ $op\gets \stat[tid]$;}
		\If{($\neg(op.\type = \treme \wedge op.\phase = phase)$)}
		\State{$return$;}
        \EndIf\If {$(\isMarked(op.\vsrc)$ $\vee$ $\isMarked(op.$  $\vdest))$}  \label{lin:hreme-loce8:16}
	    \State{$\ODA$ $fail$ $\gets new$ $\ODA$ $(phase,$ $\failure);$}
         \State{$\CAS(\stat[tid], op, fail)$;} 
        \State{$return$;}
        \EndIf
        %\State{$enode$ $e_1\gets op.ne1$;}
        \State{$enode$ $e_2\gets op.\enod$;}
        %\State{$enode$ $e_3\gets e_2.\enext$;}
        \State{$\langle pred,curr\rangle$ $\gets \nwloce(op.\vsrc,e_2.$ $\ekey)$;}
        \State{$cnext \gets curr.\enext$;}
        \If{($curr.\ekey \neq e_2.\ekey $)}
	    \State{$\ODA$ $fail$ $\gets new$ $\ODA$ $(phase,$ $\failure);$}
        \If{($\CAS(\stat[tid], op, fail)$)}
        \State{$return$;}
        \EndIf
        \Else
        %\State{$Node$ $next\gets op.curr.\vnext$;}
        \If{$(\neg \CAS(curr.\enext, cnext, \MarkedRef(c$ $next))) $}\label{lin:cas-wfreme-lgic}
       	 \State{$goto$ \lineref{search-again-hreme};}
		\EndIf
		\If{$(\neg \CAS(pred.\enext, curr,cnext)))$}\label{lin:cas-wfreme-phy}
       	 \State{$goto$ \lineref{search-again-hreme};}
		\EndIf
		 %\State{$\langle pred,curr\rangle \gets \nwloce(op.v1, e_2.\ekey)$;} {// physically removed}
	    \State{$\ODA$ $succ$ $\gets new$ $\ODA$ $(phase,$ $\success);$}
        \If{($\CAS(\stat[tid], op, succ)$)}
        \State{$return$;}
        \EndIf
        \EndIf
        \EndWhile
	\EndProcedure\label{hremeend}
	\algstore{hreme}
    \end{algorithmic}
    \hrule
    \end{subfigure}
\begin{subfigure}[t]{.5\textwidth}
\begin{algorithmic}[1]
    \renewcommand{\algorithmicprocedure}{\textbf{Operation}}
	\algrestore{hreme}
	\scriptsize
	\Procedure{$\hcone$($tid,phase$)}{}\label{hconestart}
		\State{$\ODA$ $op$ $\gets \stat[tid]$;}
		\If{($\neg(op.\type = \tcone \land op.\phase = phase)$)}
		\State{$return$;}
        \EndIf
        \State{$\enode$ $e_2\gets op.\enod$;}
        \State{$\langle pred,curr\rangle$ $\gets \nwloce(op.\vsrc,e_2.\ekey$ $)$;}
        \If{($(curr.\ekey = e_2.\ekey)$ $\wedge$ $\neg \isMarked$ $(curr.\enext)$)}
	    \State{$\ODA$ $succ$ $\gets new$ $\ODA$ $(phase,$ $\success);$}
        \State{($\CAS(\stat[tid], op, succ)$);}
        \State{$return$ $\tru$;}
        \Else
	    \State{$\ODA$ $fail$ $\gets new$ $\ODA$ $(phase,$ $\failure);$}
        \State{($\CAS(\stat[tid], op, fail)$);}
        \State{$return$ $\fal$;}
        \EndIf
	\EndProcedure\label{hconeend}
	\algstore{hcone}
\end{algorithmic}
\hrule
\begin{algorithmic}[1]
    \renewcommand{\algorithmicprocedure}{\textbf{Operation}}
	\algrestore{hcone}
	\scriptsize
	\Procedure{$\wohe$($key_1,key_2$)}{}\label{wohestart}
	    \State {$\vnode$ $v_1,$ $v_2$;}
	    \State {$\enode$ $e$;}
		\State{$\langle v_1,v_2,flag \rangle \gets LocateUV (key_1, key_2$);}
        \If {($flag$ = $\fal$)}
        \State {$return$ \vntp; }\hspace{.4cm}
        \EndIf
        \If {$(\isMarked(v_1)$ $\vee$ $\isMarked(v_2))$}  \label{lin:wohe-loce8:16}
        \State {$return$ \vntp; }
        \EndIf
		\State {$e  \gets \eh$;}
		\While{$((e.\ekey < key_2))$}
		\State {$e \gets \unMarkedRef(e.\enext)$ ;}
		\EndWhile
		\If{($(e.\ekey = key_2)$ $\land$ $(\neg\isMarked(e.$ $\enext)$ $\land$
		    $(\neg\isMarked(v_1.\vnext)$ $\land$
		    $(\neg\isMarked($ $v_2.\vnext))$)}
		 \State{ return \ep;}
		 \Else
		 \State{  return \ventp;}
		 \EndIf
	\EndProcedure\label{woheend}
	\algstore{wohe}
\end{algorithmic}
\hrule
	\end{subfigure}
	\setlength{\belowcaptionskip}{-15pt}
	\caption{Pseudo-codes of \wfaddv, \wfremv, $\wfconv$ and \wfcone}\label{fig:wf-methods4}
\end{figure*}

%% file: code/code5.tex
\begin{figure*}[]
	\begin{subfigure}[t]{.5\textwidth}
	
\begin{algorithmic}[1]
    \renewcommand{\algorithmicprocedure}{\textbf{Operation}}
	\algrestore{wohe}
	\scriptsize
	\Procedure{$\wlocv(key$)}{}\label{wlocvstart}
		%\State{$bool$ $flag$ ;}
		%\State{$\vnode$ $v3$  ;}
		\While{($\tru$)} \label{lin:search_again-vertex}
		\State {$v_1 \gets \vh$; $v_2 \gets v_1.\vnext$;} \label{lin:locv3}
	%	\State {$v_2 \gets v_1.\vnext$;} \label{lin:locv4}
		%\While{$(\tru)$}
	%	\State{$v_3 \gets v_2.\vnext$;}
	  %  \State{/*Find left and right vnode */}
	%	\While{$((read(v_2.marked) = \fal ) \land (read(v_2.val) < key))$} \label{lin:locv5} 
	    \While{($\tru$)}
	    \State {$v_3 \gets v_2.\vnext$;}
	    \While{$(\isMarked(v_3))$}
	    \State{$flag \gets \CAS(v_1.\vnext, v_2, v_3)$;}
	    %\State{$ODA$ $op\gets state[tid]$;}
	    %\If{$(op.type = OpType.Success \vee op.type = OpType.Failure)$}
		%\State {$return$ $\langle v_1,v_2 \rangle$}; {// to ensure wait-freedom} 
		%\EndIf
		\If{($\neg flag$)} 
		\State $goto$ \lineref{search_again-vertex};
		\EndIf
		\State {$v_2 \gets v_3;$ $v_3 \gets v_2.\vnext$;}
	    \EndWhile
	    \If{$(v_2.\vkey >= key)$}
		\State {$return$ $\langle v_1,v_2 \rangle$}; 
		\EndIf
	%	\State {$v_1 \gets v_2$;} \label{lin:locv6}
	%	\State {$v_2 \gets read(v_2.\vnext)$;} \label{lin:locv7}
	    \State {$v_1 \gets v_2;$ $v_2 \gets v_3$;}
		\EndWhile
		\EndWhile
	\EndProcedure\label{wlocvend}
	\algstore{hreme}
    \end{algorithmic}
    \hrule
\begin{algorithmic}[1]
    \renewcommand{\algorithmicprocedure}{\textbf{Operation}}
	\algrestore{hreme}
	\scriptsize
	\Procedure{$\nwloce(srcv,key$)}{}\label{nwocestart}
        \While{($\tru$)} \label{lin:loce-search_again}
		%	\State{/* Helping for search edge, is a supporting method for locate edge. It locates the vertices $v_1$ \& $v_2$*/}
		\State {$e_1 \gets srcv.\enext;$ $ e_2 \gets e_1.\enext$;}
        \While{($\tru$)} \label{lin:loce-search_again_2}
	    \State {$e_3 \gets e_2.\enext;$ $ v \gets e_2.\pointv$;}
	    %\State{$retry2 : $}
	    \While{$(\isMarked(v)$ $\land$ $\neg\isMarked(e_3))$}
	    \If{$(\neg\CAS(e_2.\enext,e_3,\MarkedRef(e_3)))$}
		\State {$goto$ Line \ref{lin:loce-search_again}}; 
		\EndIf
	    \If{$(\neg\CAS(e_1.\enext,e_2,e_3))$}
		\State {$goto$ Line \ref{lin:loce-search_again}}; 
		\EndIf
		\State{$e_2 \gets \unMarkedRef(e_3)$;}
		\State{$e_3 \gets e_2.\enext;$ $v \gets e_2.\pointv$;}
		\EndWhile
	    \While{$(\isMarked(e_3))$}
	    \If{$(\neg\CAS(e_1.\enext,e_2,e_3))$}
		\State {$goto$ Line \ref{lin:loce-search_again}}; 
		\EndIf
	    \State{$e_2 \gets \unMarkedRef(e_3)$;}
		\State{$e_3 \gets e_2.\enext;$ $v \gets e_2.\pointv$}
		\EndWhile
		\If{$(\isMarked(v))$}
		\State {$goto$ Line \ref{lin:loce-search_again_2}};
	    \EndIf
	    \If{$(e_2.\ekey >= key)$}
	    \State{$return \langle e_1, e_2 \rangle;$}
	    \EndIf
	    \State{$e_1 \gets e_2;$ $e_2 \gets e_3$;}
	    \EndWhile
		\EndWhile
	\EndProcedure\label{nwoceend}
	\algstore{nwoce}
\end{algorithmic}
\hrule

    \end{subfigure}
\begin{subfigure}[t]{.5\textwidth}
\begin{algorithmic}[1]
    \renewcommand{\algorithmicprocedure}{\textbf{Operation}}
	\algrestore{nwoce}
	\scriptsize
	\Procedure{$\locuv(key_1,key_2$)}{}\label{locuvstart}
        \If{$(key_1 < key_2)$}
        \State {$v_1 \gets \vh$;}\hspace{.5cm}%\label{lin:loce3} 
        \While{$(v_1.\vkey < key_1)$} %\label{lin:loce4}
		\State {$v_1 \gets v_1.\vnext$;} %\label{lin:loce5}
		\EndWhile
        \If{($v_1.\vkey$ $\neq$ $key_1$ $\vee$ $\isMarked(v_1)$)} \label{lin:helploce8}\hspace{.5cm}         \State {$return$ $\langle v_1,v_2,\fal \rangle$;}
        \EndIf
        \State {$v_2 \gets v_1.\vnext$;} %\label{lin:loce10}
		\While{$(v_2.\vkey < key_2)$} %\label{lin:loce11}
		\State {$v_2 \gets v_2.\vnext$;} %\label{lin:loce12}
		\EndWhile
        \If{($v_2.\vkey$ $\neq$ $key_2$ $\vee$ $\isMarked(v_2)$)}\label{lin:helploce16}\hspace{.5cm}
         %\label{lin:loce15}
        \State {$return$ $\langle v_1,v_2,\fal \rangle$;}
        \EndIf
        %\State {$return$;}
        \Else 
        \State {$v_2 \gets \vh$;}\hspace{.5cm}%\label{lin:loce10.1}
		\While{$(v_2.\vkey < key_2)$} %\label{lin:loce11.1}
		\State {$v_2 \gets v_2.\vnext$;} %\label{lin:loce12.1}
		\EndWhile
        \If{($v_2.\vkey$ $\neq$ $key_2$ $\vee$ ($\isMarked(v_2)$))}\label{lin:helploce16.1}\hspace{.5cm}
        \State {$return$ $\langle v_1,v_2,\fal \rangle$;}
        \EndIf
        \State {$v_1 \gets n_v.\vnext$;}%\label{lin:loce3.1}
        \While{$(v_1.\vkey < key_1)$} %\label{lin:loce4.1}
		\State {$v_1 \gets v_1.\vnext$;} %\label{lin:loce5.1}
		\EndWhile
        \If{($v_1.\vkey$ $\neq$ $key_1$ $\vee$ ($\isMarked(v_1)$))}\label{lin:helploce8.1}\hspace{.5cm}         
        \State {$return$ $\langle v_1,v_2,\fal \rangle$;}
       % \State {$return$;}
        \Else
        \State {$return$ $\langle v_1,v_2,\tru \rangle$;}
        \EndIf
        \EndIf

	\EndProcedure\label{locuvend}
	\algstore{hcone}
\end{algorithmic}
\hrule
	\end{subfigure}
	\setlength{\belowcaptionskip}{-15pt}
	\caption{Pseudo-codes of \wfaddv, \wfremv, $\wfconv$ and \wfcone}\label{fig:wf-methods5}
\end{figure*}

%% file: figs/false.pdf_t
\begin{picture}(0,0)%
\includegraphics{figs/false.pdf}%
\end{picture}%
\setlength{\unitlength}{4144sp}%
\begingroup\makeatletter\ifx\SetFigFont\undefined%
\gdef\SetFigFont#1#2#3#4#5{%
  \reset@font\fontsize{#1}{#2pt}%
  \fontfamily{#3}\fontseries{#4}\fontshape{#5}%
  \selectfont}%
\fi\endgroup%
\begin{picture}(9802,3256)(1381,-2243)
\put(1396,-331){\makebox(0,0)[lb]{\smash{{\SetFigFont{14}{16.8}{\rmdefault}{\bfdefault}{\updefault}{\color[rgb]{0,0,0}$T_1$}%
}}}}
\put(1891,-1366){\makebox(0,0)[lb]{\smash{{\SetFigFont{14}{16.8}{\rmdefault}{\bfdefault}{\updefault}{\color[rgb]{0,0,0}$T_2$}%
}}}}
\put(6031,-2041){\makebox(0,0)[lb]{\smash{{\SetFigFont{14}{16.8}{\rmdefault}{\bfdefault}{\updefault}{\color[rgb]{0,0,0}$T_3$}%
}}}}
\put(5626,794){\makebox(0,0)[lb]{\smash{{\SetFigFont{14}{16.8}{\rmdefault}{\bfdefault}{\updefault}{\color[rgb]{0,0,0}$\wfadde(u, v, \success)$}%
}}}}
\put(3286,-1096){\makebox(0,0)[lb]{\smash{{\SetFigFont{12}{14.4}{\rmdefault}{\bfdefault}{\updefault}{\color[rgb]{1,0,0}$\remv(u, \success)$}%
}}}}
\put(6661,-1726){\makebox(0,0)[lb]{\smash{{\SetFigFont{12}{14.4}{\rmdefault}{\bfdefault}{\updefault}{\color[rgb]{0,0,1}$\addv(v, \success$}%
}}}}
\put(7426, 29){\makebox(0,0)[lb]{\smash{{\SetFigFont{14}{16.8}{\sfdefault}{\mddefault}{\updefault}{\color[rgb]{0,0,0}$read(v.val)$e}%
}}}}
\put(9856, 29){\makebox(0,0)[lb]{\smash{{\SetFigFont{14}{16.8}{\sfdefault}{\mddefault}{\updefault}{\color[rgb]{0,0,0}$read(v.val)$}%
}}}}
\put(2161,-16){\makebox(0,0)[lb]{\smash{{\SetFigFont{14}{16.8}{\sfdefault}{\mddefault}{\updefault}{\color[rgb]{0,0,0}$read(u.val)$}%
}}}}
\end{picture}%

%% file: owf-algo.tex
In this section, we present the optimized version of our \wf concurrent graph \ds, which is designed based on the \fpsp algorithm by Kogan et al. \cite{Kogan+:fpsp:ppopp:2012}. The \fpsp algorithm is a combination of two parts, the first part is a \lf algorithm which is usually fast and the second one is a \wf algorithm which is slow. Pragmatically the lock-free algorithms are fast as compare to the wait-free algorithms as they don't require helping always. So, to enhance the performance and achieve a fast \wf graph we adopted \lf graph by Chatterjee et al. \cite{Chatterjee+:NbGraph:ICDCN-19} and \wf graph which described in previous \secref{wf-algo}. 
The basic working principle of the optimized \wf graph is as follows: (1). Before an operation begins the fast-path \lf algorithm, it inspects whether help needed for any other operations in the slow-path \wf algorithm, (2). Then the operation starts running with its fast-path \lf algorithms while it keeps tracking the number times it fails, which is nothing but the number of failed \CAS. Generally, if very less number of \CAS failed happen then helping is not necessary and hence the execution finishes just after completing the fast-path \lf algorithm. (3). If the operation is unable to finish its execution after trying a certain number of \CAS, then it allows entering the slow-path algorithm to finish the execution.%( the slow-path \wf algorithms described in the \secref{wf-algo}). 

The slow-path \wf algorithms described in \secref{wf-algo}. Each operation chooses its \phase numbers then it publish the operation in the \state array by updating its corresponding entry. Then it traverse through the \state array and try to help the operations whose \phase number is lower than or equals to its own \phase number, which ensures the unfinished operations gets help from other threads to finish the execution. This ensures wait-freedom. The maximum number of tries in the fast-path is upper bounded by a macro \maxfail similar to \cite{Kogan+:fpsp:ppopp:2012}, we choose \maxfail to $20$. On average, we achieve high throughput with this upper bound value. The optimized \wf algorithm is given in \figref{wf-fpsp-methods1}.% Due to lack of space, the detailed implementation described in the technical report.%\footnote{The technical report is available: https://arxiv.org/abs/1810.13325}.
%The algorithms of the optimized \wf graph given in \figref{wf-methods5} and \ref{fig:wf-methods7}.

%The idea behind the \fpsp \cite{Kogan+:WFQue:ppopp:2011} approach is combination of two parts, the first part is a lock-free algorithm which is usually fast and the second one is a wait-free algorithm which is slow. We use the Harris’s lock-free linked-list\cite{Harris:NBList:disc:2001} approach for the fast-path. The fast-path algorithm begins by a check whether helping is required for any operation in the slow-path. The operation start running with its lock-free algorithm while counting the number of contentions that end with a failed \CAS. Generally, very less number of failures occur and helping not required, and so the execution terminates after running the faster lock-free algorithm. If this fast-path fails to make progress the execution moves to the slow-path, which runs the slower wait-free algorithm described in the above \secref{graph-methods}. So requesting help using \state and making sure the operation eventually terminates. The number of \CAS failures allowed in the fast-path is limited by a parameter called \maxfail. The help is provided by threads running \fpsp ensures wait-freedom. The full implementation of the fast-path-slow-path variation of the \cgds is described in the technical report.

%% file: proof.tex
In this section, we prove the correctness of our concurrent \wf graph \ds based on  $\lp$\cite{Herlihy+:lbty:TPLS:1990} events inside the execution interval of each of the operations.
 %Here we show that the graph method
 
\begin{theorem}\normalfont The concurrent \wf graph operations are linearizable.
\end{theorem}

\begin{proof}
Based on the return values of the operations we discuss the \lp{s}.
\input{code/lps}
\noindent
From the above discussion one can notice that each operation's \lp{s} lies in the interval between the invocation and the return steps. For any invocation of a $\wfaddv(\vkey)$ operation the traversal terminates at the \vnode whose key is just less than or equal to $\vkey$. Similar reasoning also true for invocation of an $\wfadde(k_1, k_2)$ operation. Both the operations do the traversal in the sorted \vlist and \elist to make sure that a new \vnode or \enode does not break the invariant of the \wf graph \ds. The \wfremv and \wfreme do not break the sorted order of the lists. Similarly, the non-update operations do not modify the \ds. Thus we concluded that all \wf graph operations maintain the invariant across the \lp{s}. This completes the proof of the \lbty.
\end{proof}

\begin{theorem}\normalfont
	The presented concurrent graph operations $\wfaddv$, $\wfremv$, $\wfconv$, $\wfadde$, $\wfreme$, and $\wfcone$ are \wf.
\end{theorem}

\begin{proof}
To show the concurrent graph algorithms to be wait-freedom, we have to make sure that the helping procedure terminates with a limited number try in concurrent with update operations. 
As we discussed in \secref{wf-algo} that each operation chooses its \phase numbers larger than all the previous operations and then publishes the operation in the \state array by updating its entry. After that, it traverses through the \state array and tries to help the operations whose \phase number is lower than or equal to its \phase number, which ensures the unfinished operations gets help from other thread to finish the execution. So that all the threads help the pending operations and finished the execution with a limited number of steps. This ensures wait-freedom. Therefore,  the graph operations \wfaddv, \wfremv, \wfconv, \wfadde, \wfreme, and \wfcone are \wf. %The detail proof can be found the technical report.
\end{proof}
%\vspace{-2mm}
\ignore{
\begin{proof}
%If the set of keys is finite, then the size of the \wf graph has a fixed upper bound. This means there is an only finite number of \vnodes in between sentinel \vh  and \vt in the \vlist. If any non-faulty thread invokes an \wfconv, it terminates on reaching \vt with a finite number of steps. A similar argument can be brought for an \wfcone operation. Also one can see that \scrcbl operation never returns true with concurrent update operations and it terminates with a finite number of steps. This shows the \ref{lflem1_a}.

Whenever an insertion or a deletion operation is blocked by a concurrent delete operation by the way of a marked pointer, then that blocked operations is helped  to make a safe return. Generally, insertion and lookup operations do not need help by a concurrent operation. So, if any random concurrent execution consists of any concurrent \ds operation, then at least one operation finishes its execution in a finite number of steps taken be a non-faulty thread. Therefore,  the \wf graph operations \wfaddv, \wfremv, \wfconv, \wfadde, \wfreme, and \wfcone are \lf.

\end{proof}
}

%% file: code/lps.tex
\begin{enumerate}
	\item $\wfaddv(\vkey)$: We have two cases:
	\begin{enumerate}
		\item \tru: The \lp be the successful \CAS execution at the \lineref{helpaddv-succ-phyaddv-cas}.
		\item \fal: The \lp be the atomic read of the \vnext pointing to the vertex $v(\vkey)$. 
	\end{enumerate}
	\item $\wfremv(\vkey)$: We have two cases:
	\begin{enumerate}
		\item \tru: The \lp be the successful \CAS execution at the \lineref{helpremv-succ-logical-cas} (logical removal).
		\item \fal: If there is a concurrent \wfremv operation $op$, that removed $v(\vkey)$ then the \lp be just after the \lp of $op$. If $v(\vkey)$ did not exist in the \vlist then the \lp be at the invocation of $\wfremv$.\label{step:wfremv-fal} 
	\end{enumerate}
	\item $\wfconv(\vkey)$: We have two cases:
	\begin{enumerate}
		\item \tru: The \lp be the atomic read of the \vnext pointing to the vertex $v(\vkey)$.
		\item \fal: The \lp be the same as returning \fal \wfremv, the case \ref{step:wfremv-fal}.
	\end{enumerate}

	\item $\wfadde(k_1, k_2)$: We have three cases:
	\begin{enumerate}
		\item \eadd: \label{step:aceadd}
		\begin{enumerate}
			\item With no concurrent $\wfremv(k_1)$ or $\wfremv(k_2)$: The \lp be the successful \CAS execution at the \lineref{helpadde-succ-wfadde-cas}.
			\item With concurrent $\wfremv(k_1)$ or $\wfremv(k_2)$: The \lp be just before the first remove's \lp. 
		\end{enumerate}
		\item \eap: 
		\begin{enumerate}
			\item With no concurrent $\wfremv(k_1)$ or $\wfremv(k_2)$: The \lp be the atomic read of the \enext pointing to the \enode $e(k_2)$ in the \elist fo the vertex $v(k_1)$. 
			\item With concurrent $\wfremv(k_1)$ or $\wfremv(k_2)$ or $\wfreme(k_1, k_2)$ : The \lp be just before the first remove's \lp. 
		\end{enumerate}		
		\item \vntp: \label{step:wfadde-vntp}
		\begin{enumerate}
			\item At the time of invocation of $\wfadde(k_1, k_2)$ if both vertices $v(k_1)$ and $v(k_2)$ were in the \vlist and a concurrent \wfremv removes $v(k_1)$ or $v(k_2)$ or both then the \lp be the just after the \lp of the earlier \wfremv.
			\item At the time of invocation of $\wfadde(k_1, k_2)$ if both vertices $v(k_1)$ and $v(k_2)$ were not present in the \vlist, then the \lp be the invocation point itself.
		\end{enumerate}
	\end{enumerate}	
					
	\item $\wfreme(k_1, k_2)$:  We have three cases:
	\begin{enumerate}
		\item \er: 
		\begin{enumerate}
			\item With no concurrent $\wfremv(k_1)$ or $\wfremv(k_2)$: The \lp be the successful \CAS execution at the \lineref{cas-wfreme-lgic}(logical removal).
			\item With concurrent $\wfremv(k_1)$ or $\wfremv(k_2)$: The \lp be just before the first remove's \lp. 
		\end{enumerate}
		\item \entp: If there is a concurrent \wfreme operation removed $e(k_1,k_2)$ then the \lp be the just after its \lp, otherwise at the invocation of $\wfreme(k_1, k_2)$ itself. 	
		\item \vntp: The \lp be the same as the case \wfadde returning ``\vntp''\ref{step:wfadde-vntp}.	
	\end{enumerate}			
	
	\item $\wfcone(k_1, k_2)$: Similar to \wfreme, we have three cases:
	\begin{enumerate}
		\item \ep:
		\begin{enumerate}
			\item With no concurrent $\wfremv(k_1)$ or $\wfremv(k_2)$: The \lp be the atomic read of the \enext pointing to the \enode $e(k_2)$ in the \elist fo the vertex $v(k_1)$. 
			\item With concurrent $\wfremv(k_1)$ or $\wfremv(k_2)$ or $\wfreme(k_1, k_2)$ : The \lp be just before the first remove's \lp.
		\end{enumerate}
		\item \vntp: The \lp be the same as that of the \wfadde's returning ``\vntp'' case \ref{step:wfadde-vntp}.	
		\item \ventp: The \lp be the same as that of the \wfreme's returning ``\entp'' and \wfadde's returning ``\vntp'' cases.
	\end{enumerate}

\end{enumerate}

%% file: expt-results.tex
\input{code/result1}

\textbf{Experimental Setup:}
We conducted our experiments on a processor with Intel(R) Xeon(R) E5-2690 v4 CPU containing 56 cores running at 2.60GHz. Each core supports 2 logical threads. Every core's L1-64K, L2-256K cache memory is private to that core; L3-35840K cache is shared across the cores, 32GB of RAM and 1TB of hard disk, running 64-bit Linux operating system. All the implementation\footnote{The source code is available on https://github.com/PDCRL/ConcurrentGraphDS.} were written in C++ (without any garbage collection) and multi-threaded implementation is based on Posix threads. 

\textbf{Running Strategy:}
In the experiments, we start with an initial graph of $1000$ vertices and nearly $\binom{1000}{2}/4 \approx 125000$ edges added randomly, rather than starting with an empty graph. When the program begins, it creates a fixed set of threads (1, 10, 20, 30, 40, 50, 60 and 70) and each thread randomly performs a set of operations chosen by a particular workload distribution, as given below. The metric for evaluation is the number of operations completed in a unit time, i.e. throughput. We measured throughput by running the experiment for 20 seconds. Each data point is obtained by averaging over 5 iterations. 

\textbf{Workload Distribution:}
To compare the performance with several micro benchmarks, we used the following distributions over the ordered set of operations $\{$\wfaddv, \wfremv, \wfconv, \wfadde, \wfreme, \wfcone$\}$: %(1). \textit{Lookup Intensive}: $($$2.5\%$, $2.5\%$, $45\%$, $2.5\%$, $2.5\%$, $45\%$$)$, see the \figref{subfig9010}, (2). \textit{Equal Lookup and Updates}: $($$12.5\%$, $12.5\%$, $25\%$, $12.5\%$, $12.5\%$, $25\%$$)$, see the \figref{subfig5050}, and (3). \textit{Update Intensive}: $($$22.5\%$, $22.5\%$, $5\%$, $22.5\%$, $22.5\%$, $5\%$$)$, \figref{subfig1090}.
%\ignore{
\begin{enumerate}
    \item \textit{Lookup Intensive}: $($$2.5\%$, $2.5\%$, $45\%$, $2.5\%$, $2.5\%$, $45\%$$)$, see the \figref{subfig9010}.
    \item \textit{Equal Lookup and Updates}: $($$12.5\%$, $12.5\%$, $25\%$, $12.5\%$, $12.5\%$, $25\%$$)$, see the \figref{subfig5050}. 
    \item \textit{Update Intensive}: $($$22.5\%$, $22.5\%$, $5\%$, $22.5\%$, $22.5\%$, $5\%$$)$, \figref{subfig1090}.
   % \item \textit{Update Only}: $($$25\%$, $25\%$, $0\%$, $25\%$, $25\%$, $0\%$$)$, \figref{subfig10900}    
\end{enumerate}
%}

% The tests were performed in a controlled environment, where we were the sole users of the system.

We compare the following concurrent graph algorithms:
\begin{enumerate}
    \item \textbf{Seq}: Sequential execution of all the operations.
    \item \textbf{Coarse}: Execution with a coarse grained lock \cite[Ch. 9]{Maurice+:AMP:book:2012}.
    \item \textbf{HoH}: Execution with Hand-over-Hand lock \cite[Ch. 9]{Maurice+:AMP:book:2012}.
    \item \textbf{Lazy}: Execution with Lazy-lock \cite{Heller+:LazyList:PPL:2007}.
    \item\textbf{NBGraph}: Based on \nbk graph  \cite{Chatterjee+:NbGraph:ICDCN-19}. 
    \item \textbf{WFGraph-woh}: The \wf graph algorithm with \wohv and \wohe (without helping of contains vertex and edge operations) \secref{wf-algo}. 
    \item \textbf{WFGraph-wh}: The \wf graph algorithm with \wfconv and \wfcone(with helping of contains vertex and edge operation) \secref{wf-algo}. 
    \item \textbf{OWFGraph-woh}: The optimized version of \wf graph algorithm with \wohv and \wohe(without helping of contains vertex and edge operations) \secref{fpsp-algo}.
    \item \textbf{OWFGraph-wh}: The optimized version of \wf graph algorithm with \wfconv and \wfcone(with helping of contains vertex and edge operation) \secref{fpsp-algo}.
\end{enumerate}

In the plots shown in \figref{cgds}, we observe that the \texttt{WFGraph-woh} and \texttt{WFGraph-wh} algorithm does not scale well like \texttt{NBGraph} with the number of threads in the system, and saturate at $56$ threads(number of cores), on the other hand, the \texttt{OWFGraph-woh} variant scales well compare to others.  
%except \figref{subfig10900}(the case having update only). 

% as explained in 

%% file: code/result1.tex
\begin{figure*}[!hbtp]
\captionsetup{font=scriptsize}
    \begin{subfigure}[b]{0.28\textwidth}
    \setlength{\belowcaptionskip}{-2mm}
    \captionsetup{font=scriptsize, justification=centering}
    \caption{Lookup Intensive }
        \centering
        \resizebox{\linewidth}{!}{
	\begin{tikzpicture} [scale=0.3]
	\begin{axis}[legend style={at={(0.5,1)},anchor=north},
	xlabel=No. of threads,
	ylabel=Throughput ops/sec, ylabel near ticks]
	%\addplot table [x=Threads,  y=$Seq$]{results/WFG10U90L.dat};
	\addplot table [x=Threads,  y=$Coarse$]{results/WFG10U90L.dat};
	\addplot table [x=Threads, y=$HoH$]{results/WFG10U90L.dat};
	\addplot table [x=Threads, y=$Lazy$]{results/WFG10U90L.dat};
	\addplot table [x=Threads, y=$NBGraph$]{results/WFG10U90L.dat};
	\addplot table [x=Threads, y=$WFG-woh$]{results/WFG10U90L.dat};
	\addplot table [x=Threads, y=$OWFG-woh$]{results/WFG10U90L.dat};
	\addplot table [x=Threads, y=$WFG-wh$]{results/WFG10U90L.dat};
	\addplot table [x=Threads, y=$OWFG-wh$]{results/WFG10U90L.dat};
	\addplot[blue,sharp plot,update limits=false] 
	coordinates {(-15,194103) (90,194103)} 
	node[above] at (axis cs:55,194103) {base-line(Seq)};
	%	\addlegendentry{$VElock-DIE$}
	\end{axis}
	\end{tikzpicture}
        }
        \label{fig:subfig9010}
    \end{subfigure}
    \begin{subfigure}[b]{0.28\textwidth}
    \setlength{\belowcaptionskip}{-2mm}
    \captionsetup{font=scriptsize, justification=centering}
    \caption{Equal Lookup and Updates}   
    \centering
        \resizebox{\linewidth}{!}{
   	\begin{tikzpicture}[scale=0.30]
	\begin{axis}[
	xlabel=No. of threads,
	ylabel=Throughput ops/sec, ylabel near ticks]
	%\addplot table [x=Threads,  y=$Seq$]{results/WFG50U50L.dat};
	\addplot table [x=Threads,  y=$Coarse$]{results/WFG50U50L.dat};
	\addplot table [x=Threads, y=$HoH$]{results/WFG50U50L.dat};
	\addplot table [x=Threads, y=$Lazy$]{results/WFG50U50L.dat};
	\addplot table [x=Threads, y=$NBGraph$]{results/WFG50U50L.dat};
	\addplot table [x=Threads, y=$WFG-woh$]{results/WFG50U50L.dat};
	\addplot table [x=Threads, y=$OWFG-woh$]{results/WFG50U50L.dat};
	\addplot table [x=Threads, y=$WFG-wh$]{results/WFG50U50L.dat};
	\addplot table [x=Threads, y=$OWFG-wh$]{results/WFG50U50L.dat};
	\addplot[blue,sharp plot,update limits=false] 
	coordinates {(-15,193894) (90,193894)} 
	node[above] at (axis cs:35,193894) {base-line(Seq)};
%	\addlegendentry{$VElock-DIE$}
	\end{axis}
	\end{tikzpicture}
        }
        \label{fig:subfig5050}
    \end{subfigure}
    \begin{subfigure}[b]{0.39\textwidth}
    \captionsetup{font=scriptsize, justification=centering}
    \setlength{\belowcaptionskip}{-2mm}
    \caption{Update Intensive}
        \centering
        \resizebox{\linewidth}{!}{
           \begin{tikzpicture}[scale=0.30]
	%\caption{M1}style={at={(0.5,0.6)},anchor=north}
	\begin{axis}[legend pos=outer north east ,
	xlabel=No. of threads,
	ylabel=Throughput ops/sec, ylabel near ticks]
	%\addplot table [x=Threads,  y=$Seq$]{results/WFG90U10L.dat};
%	\addlegendentry{Seq}
	\addplot table [x=Threads,  y=$Coarse$]{results/WFG90U10L.dat};
	\addlegendentry{Coarse}
	\addplot table [x=Threads, y=$HoH$]{results/WFG90U10L.dat};
	\addlegendentry{HoH}
	\addplot table [x=Threads, y=$Lazy$]{results/WFG90U10L.dat};
	\addlegendentry{Lazy}
	\addplot table [x=Threads, y=$NBGraph$]{results/WFG90U10L.dat};
	\addlegendentry{NBGraph}
	\addplot table [x=Threads, y=$WFG-woh$]{results/WFG90U10L.dat};
		\addlegendentry{WFGraph\text{-}woh}
	\addplot table [x=Threads, y=$OWFG-woh$]{results/WFG90U10L.dat};
	\addlegendentry{OWFGraph\text{-}woh}
	\addplot table [x=Threads, y=$WFG-wh$]{results/WFG90U10L.dat};
		\addlegendentry{WFGraph\text{-}wh}
	\addplot table [x=Threads, y=$OWFG-wh$]{results/WFG90U10L.dat};
	\addlegendentry{OWFGraph\text{-}wh}
	\addplot[blue,sharp plot,update limits=false] 
	coordinates {(-15,193926) (90,193926)} 
	node[above] at (axis cs:55,193926) {base-line(Seq)};
	
	\end{axis}
	\end{tikzpicture}
        }
        \label{fig:subfig1090}
    \end{subfigure}
    \vspace{-4mm}
     	\setlength{\belowcaptionskip}{-15pt}
    \caption{Experimental Results for Wait-free and Optimised Wait-free Graph.}
   
    \label{fig:cgds}
    \setlength{\belowcaptionskip}{-15pt}
\end{figure*}
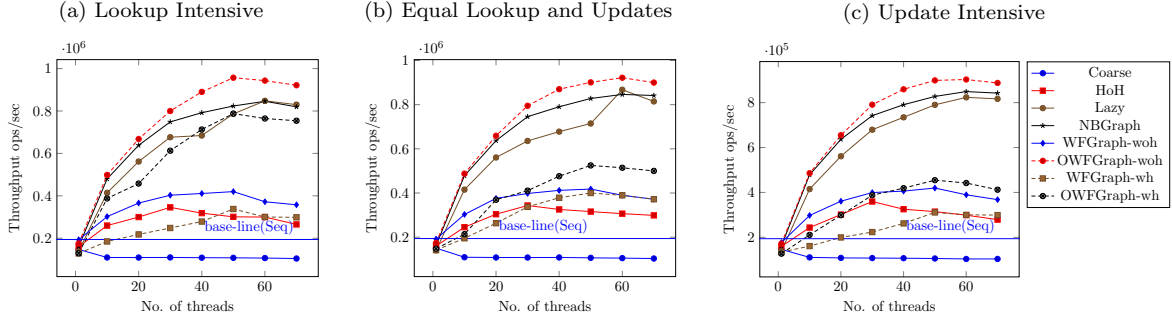

%% file: conclusion.tex
In this paper, we presented an efficient, practical \wf algorithm to implement a \cgds, which allows threads to insert and delete the vertices/edges concurrently. We also developed an optimized version of \wf graph using the concept of \fpsp algorithm developed by Kogan et al. \cite{Kogan+:fpsp:ppopp:2012}. We extensively evaluated the C++ implementation of our algorithm and the optimized variant through several micro-benchmarks. We compared \wf graph and optimized \wf graph algorithms with sequential, coarse-lock, hand-over-hand lock, lazy lock, and \nbk concurrent graphs. The optimized \wf graph without helping of contains vertex and edge operations achieves nearly up to $9$x speedup on throughput with respect to locking counterparts and nearly $1.5$x speedup with respect to \nbk counterpart. Currently, our implementation does not have any garbage collection mechanism. In future, we plan to enhance our implementation with a garbage collection procedure similar to \cite{Michael:HPS:TPDS:2004}.

%Our experiments shows that optimized \wf graph without helping of contains vertex and edge operations scales well ac compare to others.

% which can be undesirable for languages like C/C++ which does not have in-built garbage collection mechanism

%% file: code/code6.tex
\begin{figure*}[!t]
	\begin{subfigure}[t]{.5\textwidth}
\begin{algorithmic}[1]
	\algrestore{hcone}
\scriptsize
\renewcommand{\algorithmicprocedure}{\textbf{Operation}}
	\Procedure{$\mxph()$}{}\label{owfmxphstart}
	 \State{$\maxphase.\fadd(1)$;}
	 \State{$return$ $\maxphase$;}
	\EndProcedure\label{owfmxphend}
	\algstore{maxph}
\end{algorithmic}
\hrule
\begin{algorithmic}[1]
	\algrestore{maxph}
\renewcommand{\algorithmicprocedure}{\textbf{Operation}}
\scriptsize
	\Procedure{ \hgds($phase$)}{}\label{owfhgdsstart}
	    %\State{$tid \gets 0$;}
	    \For{($tid \gets 0$ $to $ $\stat.end()$)}
	    \State{$\ODA$ $desc\gets \stat[tid]$;}
	    \If{($desc.\phase \leq phase$)}
	    \If{($desc.\type = \taddv$)}
	    \State{$\haddv(tid, desc.\phase)$;}
	    \ElsIf{($desc.\type = \tremv$)}
	    \State{$\hremv(tid,desc.\phase)$;}
	    \ElsIf{($desc.\type = \tadde$)}
	    \State{$\hadde(tid,desc.\phase)$;}
	    \ElsIf{($desc.\type = \treme$)}
	    \State{$\hreme(tid,desc.\phase)$;}
	    \ElsIf{($desc.\type = \tconv$)}
	    \State{$\wfconv(tid,desc.\phase)$;}
	    \ElsIf{($desc.\type = \tcone$)}
	    \State{$\wfcone(tid,desc.\phase)$;}
	    \EndIf
	    \EndIf
	    %\State{$tid \gets tid + 1$};
	    \EndFor
	\EndProcedure\label{owfhgdsend}
	\algstore{hgds}
\end{algorithmic}	%\caption{}\label{addv}
	%\end{subfigure}% need this comment symbol to avoid overconvfull hbox
	%\begin{subfigure}{.24\textwidth}
	%\end{subfigure}
	    %\begin{subfigure}{.5\textwidth}
	    \hrule
\begin{algorithmic}[1]
    \renewcommand{\algorithmicprocedure}{\textbf{
Operation}}
	\algrestore{hgds}
	\scriptsize
	\Procedure{ $\naddv$($key$)}{}\label{owfnaddvstart}
	%\State{$tries \gets 0$};
	%\While{$(tries < \maxfail)$}
    \For{($tries \gets 0$ $to $ $\maxfail$)}
	\State{$\langle pv,cv \rangle \gets \wlocv(key)$};
	\If{$(cv.\vkey = key)$}
	\State{$return$ $\fal$};
	\Else
	\State{$\vnode$ $v$};
	\State{$\vnode$ $next$ $\gets  v.\vnext$};
	\State{$\CAS(v.\vnext, next, cv)$;}
    \If{($\CAS(pv.\vnext, curr, v)$)} \label{lin:ref-check}
    \State{$return$ $\tru$};
	\EndIf
	\EndIf
	%\State{$tries \gets tries + 1$};
	%\EndWhile
	\EndFor
	\State{$return$ $\wfaddv(key)$};
	\EndProcedure\label{owfnaddvend}
	\algstore{naddv}
    \end{algorithmic}
    \hrule
\begin{algorithmic}[1]
    \renewcommand{\algorithmicprocedure}{\textbf{Operation}}
	\algrestore{naddv}
	\scriptsize
	\Procedure{ $\nremv$($key$)}{}\label{owfnremvstart}
%		\State{$tries \gets 0$};
%		\While{$(tries < \maxfail)$}\label{lin:rem_vertex_1}
	    \For{($tries \gets 0$ $to $ $\maxfail$)} \label{lin:rem_vertex_1}
		\State{$\langle pv,cv \rangle \gets \wlocv(key)$};
		\If{$(cv.\vkey \neq key)$}
		\State{$return$ $\fal$};
		\Else
		\State{$se \gets cv.\vnext$};
		\State{$me \gets \MarkedRef(se)$};
		\If{$(\neg\CAS(cv.\vnext,se,me))$}
		\State{$goto$ \lineref{rem_vertex_1}};
		\EndIf
		\If{$(\CAS(pv.\vnext,ce,se))$}
		\State{$return$ $\tru$};
		\EndIf
		\EndIf
	%	\State{$tries \gets tries + 1$};
		\EndFor
		\State{$return$ $\wfremv(key)$};
	\EndProcedure\label{owfnremvend}
	\algstore{wfremv}
    \end{algorithmic}
    \hrule
\begin{algorithmic}[1]
    \renewcommand{\algorithmicprocedure}{\textbf{Operation}}
	\algrestore{wfremv}
	\scriptsize
	\Procedure{$\nconv$($key$)}{}\label{owfnconvstart}
%		\State{$tries \gets 0$};
%		\While{$(tries < \maxfail)$}\label{lin:rem_vertex}
	    \For{($tries \gets 0$ $to $ $\maxfail$)} 
		\State {$\vnode$ $v \gets read(\vh)$;}
		\While{$(v.\vkey < key)$}
		\State {$v \gets \unMarkedRef(v.\vnext)$ ;}
		\EndWhile
		\State {return $(v.\vkey = key) \land (IsMrkd(v) = \fal)$};
%		\State{$tries \gets tries + 1$};
		\EndFor
		\State{$return$ $\wfconv(key)$};
	\EndProcedure\label{owfnconvend}
	\algstore{wfconv}
\end{algorithmic}
\hrule
  	\end{subfigure}
\begin{subfigure}[t]{.5\textwidth}
\begin{algorithmic}[1]
    \renewcommand{\algorithmicprocedure}{\textbf{Operation}}
	\algrestore{wfconv}
	\scriptsize
	\Procedure{$\nadde$($key_1,key_2$)}{}\label{owfnaddestart}
%		\State{$tries \gets 0$};
%		\While{$(tries < \maxfail)$}\label{lin:rem_vertex}
	    \For{($tries \gets 0$ $to $ $\maxfail$)} 
		\State{$\langle v_1,v_2,flag\rangle$ $\gets$ $LocateUV(key_1,key_2$ $);$}
        \If {($flag$ = $\fal$)}
        \State {$return$ $\fal$};
        \Else
        \State{$\langle pe,ce \rangle \gets \nwloce(v_1,key_2)$;}
        \If{$(ce.\vkey = key_2)$}
        \State{$return$ $\fal$};
        \Else
        \State{$newe \gets \enode(key_2)$};
        \State{$next \gets newe.\enext$};
        \State{$\CAS(newe.\enext,next,cv)$};
        \If{$(\CAS(pe.\enext,ce,newe))$};
        \State{$newe.\pointv \gets v_2$;}
        \State{$return$ $\tru$};
        \EndIf
        \EndIf
        \EndIf
%        \State{$treis \gets tries + 1$};
		\EndFor
		\State{$return$ $\wfadde(key_1,key_2)$};
	\EndProcedure\label{owfnaddeend}
	\algstore{wfadde}
\end{algorithmic}
\hrule
\begin{algorithmic}[1]
    \renewcommand{\algorithmicprocedure}{\textbf{Operation}}
	\algrestore{wfadde}
	\scriptsize
	\Procedure{$\nreme$($key_1,key_2$)}{}\label{owfnremestart}
%		\State{$tries \gets 0$};
%		\While{$(tries < \maxfail)$}\label{lin:rem_edge}
	    \For{($tries \gets 0$ $to $ $\maxfail$)} \label{lin:rem_edge1}
	    \State{$\langle v_1,v_2,flag\rangle$ $\gets$ $LocateUV(key_1,key_2$ $);$}
        \If {($flag$ = $\fal$)}
        \State {$return$ $\fal$};
        \Else
        \State{$\langle pe,ce \rangle \gets \nwloce(v_1,key_2)$;}
        \If{$(ce.\ekey \neq key_2)$}
        \State{$return$ $\fal$};
        \Else
        \State{$se \gets ce.\enext$};
        \State{$me \gets \MarkedRef(se)$};
        \If{$(\neg\CAS(ce.\enext,se,me))$};
        \State{$goto$ \lineref{rem_edge1}};
        \EndIf
        \If{$(\CAS(pe.\enext,ce,succ))$};
        \State{$return$ $\tru$};
        \EndIf
        \EndIf
        \EndIf
        %\State{$treis \gets tries + 1$};
		\EndFor
		\State{$return$ $\wfreme(key_1,key_2)$};
	\EndProcedure\label{owfnremeend}
	\algstore{nreme}
\end{algorithmic}
\hrule
\begin{algorithmic}[1]
    \renewcommand{\algorithmicprocedure}{\textbf{Operation}}
	\algrestore{nreme}
 	\scriptsize
	\Procedure{$\ncone$($key_1,key_2$)}{}\label{owfnconestart}
%		\State{$tries \gets 0$};
%		\While{$(tries < \maxfail)$}\label{lin:rem_edge}
	    \For{($tries \gets 0$ $to $ $\maxfail$)} 
	    \State {$\vnode$ $v_1,$ $v_2$;}
	    \State {$\enode$ $e$;}
	    \State{$\langle v_1,v_2,flag\rangle$ $\gets$ $LocateUV(key_1$ $,ke$ $y_2);$}
        \If {($flag$ = $\fal$)}
        \State {$return$ \vntp; }\hspace{.4cm}
        \EndIf
        \If {$(\isMarked(v_1)$ $\vee$ $\isMarked(v_2))$}  \label{lin:wohe-loce8:16_}
        \State {$return$ \vntp; }
        \EndIf
		\State {$e  \gets \eh$;}
		\While{$((e.\ekey < key_2))$}
		\State {$e \gets \unMarkedRef(e.\enext)$ ;}
		\EndWhile
		\If{($(e.\ekey = key_2)$ $\land$ $(\neg\isMarked(e.$ $\enext)$ $\land$
		    $(\neg\isMarked(v_1.\vnext)$ $\land(\neg$ $\isMarked($ $v_2.\vnext))$)}
		 \State{ return \ep;}
		 \Else
		 \State{  return \ventp;}
		 \EndIf
%        \State{$treis \gets tries + 1$};
		\EndFor
		\State{$return$ $\wfcone(key_1,key_2)$};
	\EndProcedure \label{owfnconeend}
%	\algstore{cone}
    \end{algorithmic}
    \hrule
	\end{subfigure}
	\setlength{\belowcaptionskip}{-15pt}
	\caption{Pseudo-codes of \wfaddv, \wfremv, $\wfconv$ and \wfcone}\label{fig:wf-fpsp-methods1}
\end{figure*}